\documentclass[twoside,11pt]{article} 
\usepackage{jmlr2e}
\usepackage{hyperref}
\hypersetup{colorlinks,
            urlcolor=blue,
            citecolor=green,
            linkcolor=red}
\usepackage{url}
\usepackage{mathrsfs}
\usepackage{amsmath}
\usepackage{amsfonts}
\usepackage{algorithmic}
\usepackage{algorithm}
\usepackage{color}
\usepackage{mathtools}

\usepackage{pstricks}
\usepackage{tikz}
\usepackage{subcaption}

\newtheorem{defn}{Definition}
\newtheorem{thm}{Theorem}
\newtheorem{cor}{Corollary}

\begin{document}

\title{Plug-and-play dual-tree algorithm runtime analysis}

\author{\name Ryan R. Curtin \email ryan@ratml.org\\
  \addr School of Computational Science and Engineering\\
  Georgia Institute of Technology\\
  Atlanta, GA 30332-0250, USA
  \AND
  \name Dongryeol Lee \email drselee@gmail.com\\
  \addr GE Global Research Center\\
  Schenectady, NY 12309
  \AND
  \name William B. March \email march@ices.utexas.edu\\
  \addr Institute for Computational Engineering and Sciences\\
  University of Texas, Austin\\
  Austin, TX 78712-1229
  \AND
  \name Parikshit Ram \email p.ram@gatech.edu\\
  \addr Skytree, Inc.\\
  Atlanta, GA 30332
}

\editor{Unknown}

\maketitle

\begin{abstract}
Numerous machine learning algorithms contain pairwise statistical problems at
their core---that is, tasks that require computations over all pairs of input
points if implemented naively.  Often, tree structures are used to solve these
problems efficiently. Dual-tree algorithms can efficiently solve or approximate
many of these problems.  Using cover trees, rigorous worst-case runtime
guarantees have been proven for some of these algorithms. In this paper, we
present a {\em problem-independent} runtime guarantee for {\em any} dual-tree
algorithm using the cover tree, separating out the problem-dependent and the
problem-independent elements.  This allows us to just plug in bounds for the
problem-dependent elements to get runtime guarantees for dual-tree algorithms
for any pairwise statistical problem without re-deriving the entire proof. We
demonstrate this plug-and-play procedure for nearest-neighbor search and
approximate kernel density estimation to get improved runtime guarantees.  Under
mild assumptions, we also present the first linear runtime guarantee for
dual-tree based range search.
\end{abstract}


\section{Dual-tree algorithms}

A surprising number of machine learning algorithms have computational
bottlenecks that can be expressed as pairwise statistical problems.
By this, we mean computational tasks that can be evaluated directly by iterating over all
pairs of input points.
Nearest neighbor search is one such problem, since for every query point, we can
evaluate its distance to every reference point and keep the closest one.
This naively requires $O(N)$ time
to answer in single query in a reference set of size $N$; answering $O(N)$
queries subsequently requires prohibitive $O(N^2)$ time. Kernel density
estimation is also a pairwise statistical problem, since we compute a sum over all
reference points for each query point.
This again requires $O(N^2)$
time to answer $O(N)$ queries if done directly.
The reference set is typically indexed with
spatial data structures to accelerate this type of computation
\citep{finkel1974quad, langford2006}; these result in $O(\log N)$ runtime per
query under favorable conditions.

Building upon this intuition, \citet{nbody} generalized the fast multipole
method from computational physics to obtain dual-tree algorithms.  These are
extremely useful when there are large query sets, not just a few query points.
Instead of building a tree on the reference set and searching with each query
point separately, Gray and Moore suggest also building a query tree and
traversing both the query and reference trees simultaneously (a {\it dual-tree
traversal}, from which the class of algorithms takes its name).

Dual-tree algorithms can be easily understood through the recent framework of
\citet{curtin2013tree}: two trees (a query tree and a reference
tree) are traversed by a {\it pruning dual-tree traversal}.  This traversal
visits combinations of nodes from the trees in some sequence (each combination
consisting of a query node and a reference node), calling a problem-specific
\texttt{Score()} function to determine if the node combination can be pruned.
If not, then a problem-specific \texttt{BaseCase()} function is called for each
combination of points held in the query node and reference node.  This has
significant similarity to the more common single-tree branch-and-bound
algorithms, except that the algorithm must recurse into child nodes of {\it
both} the query tree and reference tree.

There exist numerous dual-tree algorithms for problems as diverse as kernel
density estimation \citep{gray2003nonparametric}, mean shift \citep{wang2007fast},
minimum spanning tree calculation \citep{march2010euclidean}, $n$-point
correlation function estimation \citep{march2012fast}, max-kernel search
\citep{curtin2013fast}, particle smoothing \citep{klaas2006fast}, variational
inference \citep{amizadeh2012variational}, range search \citep{nbody}, and
embedding techniques \cite{maaten2014accelerating}, to name a few.

Some of these algorithms are derived using the cover tree \citep{langford2006}, a
data structure with compelling theoretical qualities.  When cover trees are
used, Dual-tree all-nearest-neighbor search and approximate kernel density
estimation have $O(N)$ runtime guarantees for $O(N)$ queries \citep{ram2009};
minimum spanning tree calculation scales as $O(N \log N)$
\citep{march2010euclidean}.  Other problems have similar worst-case guarantees
\citep{curtin2014dual, march2013multi}.

In this work we combine the generalization of \citet{curtin2013tree} with the
theoretical results of \citet{langford2006} and others in order to develop a
worst-case runtime bound for any dual-tree algorithm when the cover tree is
used.

Section \ref{sec:preliminaries} lays out the required background, notation, and
introduces the cover tree and its associated theoretical properties.  Readers
familiar with the cover tree literature and dual-tree algorithms
\citep[especially][]{curtin2013tree} may find that section to be review.
Following that, we introduce an intuitive measure of cover tree imbalance, an
important property for understanding the runtime of dual-tree algorithms, in
Section \ref{sec:imbalance}.  This measure of imbalance is then used to prove
the main result of the paper in Section \ref{sec:bound}, which is a worst-case
runtime bound for generalized dual-tree algorithms.  We apply this result to
three specific problems: nearest neighbor search (Section \ref{sec:nns}),
approximate kernel density estimation (Section \ref{sec:akde}), and range search
/ range count (Section \ref{sec:rs}), showing linear runtime bounds for each of
those algorithms.  Each of these bounds is an improvement on the
state-of-the-art, and in the case of range search, is the first such bound.

\section{Preliminaries}
\label{sec:preliminaries}

For simplicity, the algorithms considered in this paper will be presented in a
tree-independent context, as in \citet{curtin2013tree}, but the only type of
tree we will consider is the cover tree \citep{langford2006}, and the only type
of traversal we will consider is the cover tree pruning dual-tree traversal,
which we will describe later.

As we will be making heavy use of trees, we must establish notation \citep[taken
from][]{curtin2013tree}.  The notation we will be using is defined in Table
\ref{tab:notation}.

\begin{table}
{\small
\begin{center}
\begin{tabular}{|c|l|}
\hline
{\bf Symbol} & {\bf Description} \\ \hline
$\mathscr{N}$ & A tree node \\ \hline
$\mathscr{C}_i$ & Set of child nodes of $\mathscr{N}_i$ \\ \hline
$\mathscr{P}_i$ & Set of points held in $\mathscr{N}_i$ \\ \hline
$\mathscr{D}_i^n$ & Set of descendant nodes of $\mathscr{N}_i$ \\ \hline
$\mathscr{D}_i^p$ & Set of points contained in $\mathscr{N}_i$ and
$\mathscr{D}_i^n$ \\ \hline
$\mu_i$ & Center of $\mathscr{N}_i$ (for cover trees, $\mu_i = p_i$) \\ \hline
$\lambda_i$ & Furthest descendant distance \\ \hline
\end{tabular}
\end{center}
}
\caption{Notation for trees.  See \cite{curtin2013tree} for details.}
\label{tab:notation}
\end{table}

\subsection{The cover tree}

The cover tree is a leveled hierarchical data structure originally proposed for
the task of nearest neighbor search by \citet{langford2006}.  Each node
$\mathscr{N}_i$ in the cover tree is associated with a single point $p_i$.  An
adequate description is given in their work (we have adapted notation slightly):

\begin{quote}
A {\it cover tree} $\mathscr{T}$ on a dataset $S$ is a leveled tree where each
level is a ``cover'' for the level beneath it.  Each level is indexed by an
integer scale $s_i$ which decreases as the tree is descended.  Every {\it node}
in the tree is associated with a point in $S$.  Each {\it point} in $S$ may be
associated with multiple nodes in the tree; however, we require that any point
appears at most once in every level.  Let $C_{s_i}$ denote the set of points in
$S$ associated with the nodes at level $s_i$.  The cover tree obeys the
following invariants for all $s_i$:

\begin{itemize}
  \item {\em (Nesting)}. $C_{s_i} \subset C_{s_i - 1}$.  This implies that once a
point $p \in S$ appears in $C_{s_i}$ then {\it every} lower level in the tree
has a node associated with $p$.

  \item {\em (Covering tree)}. For every $p_i \in C_{s_i - 1}$, there exists a
$p_j \in C_{s_i}$ such that $d(p_i, p_j) < 2^{s_i}$ and the node in level $s_i$
associated with $p_j$ is a parent of the node in level $s_i - 1$ associated with
$p_i$.

  \item {\em (Separation)}.  For all distinct $p_i, p_j \in C_{s_i}$, $d(p_i,
p_j) > 2^{s_i}$.
\end{itemize}
\end{quote}

As a consequence of this definition, if there exists a node $\mathscr{N}_i$,
containing the point $p_i$ at some scale $s_i$, then there will also exist a
self-child node $\mathscr{N}_{ic}$ containing the point $p_i$ at scale $s_i - 1$
which is a child of $\mathscr{N}_i$.  In addition, every descendant point of the
node $\mathscr{N}_i$ is contained within a ball of radius $2^{s_i + 1}$ centered
at the point $p_i$; therefore, $\lambda_i = 2^{s_i + 1}$ and $\mu_i = p_i$
(Table \ref{tab:notation}).

Note that the cover tree may be interpreted as an infinite-leveled tree, with
$C_{\infty}$ containing only the root point, $C_{-\infty} = S$, and all levels
between defined as above.  \citet{langford2006} find
this representation (which they call the {\it implicit} representation) easier
for description of their algorithms and some of their proofs.  But clearly,
this is not suitable for implementation; hence, there is an {\it explicit}
representation in which all nodes that have only a self-child are coalesced
upwards (that is, the node's self-child is removed, and the children of that
self-child are taken to be the children of the node).  

In this work, we consider only the explicit representation of a cover tree, and
do not concern ourselves with the details of tree construction\footnote{A batch
construction algorithm is given by \citet{langford2006}, called
\texttt{Construct}.}.

\subsection{Expansion constant}

The explicit representation of a cover tree has a number of useful theoretical
properties based on the expansion constant \citep{karger2002finding}; we restate
its definition below.

\begin{defn}
\label{def:int_dim}
Let $B_S(p, \Delta)$ be the set of points in $S$ within a closed ball of radius
$\Delta$ around some $p \in S$ with respect to a metric $d$:
$B_S(p, \Delta) = \{ r \in S \colon d(p, r) \leq \Delta \}$.
Then, the {\bf expansion constant} of $S$ with respect to the metric $d$ is the
smallest $c \ge 2$ such that

\begin{equation}
| B_S(p, 2 \Delta) | \le c | B_S(p, \Delta) |\ \forall\ p \in S,\
\forall\ \Delta > 0.
\end{equation}

\end{defn}

The expansion constant is used heavily in the cover tree literature.  It is,
in some sense, a notion of instrinic dimensionality, and previous work has shown
that there are many scenarios where $c$ is independent of the number of points
in the dataset \citep{karger2002finding, langford2006,
krauthgamer2004navigating, ram2009}.  Note also that if points in $S \subset
\mathcal{H}$ are being drawn according to a stationary distribution $f(x)$, then
$c$ will converge to some finite value $c_f$ as $|S| \to \infty$.  To see this,
define $c_f$ as a generalization of the expansion constant for distributions.
$c_f \ge 2$ is the smallest value such that

\begin{equation}
\int_{\mathcal{B}_{\mathcal{H}}(p, 2 \Delta)} f(x) dx \le c_{f}
\int_{\mathcal{B}_{\mathcal{H}}(p, \Delta)} f(x)
dx
\end{equation}

\noindent for all $p \in \mathcal{H}$ and $\Delta > 0$ such that
$\int_{\mathcal{B}_{\mathcal{H}}(p, \Delta)} f(x) dx > 0$, and with
$\mathcal{B}_{\mathcal{H}}(p, \Delta)$ defined as the closed ball of radius
$\Delta$ in the space $\mathcal{H}$.

As a simple example, take $f(x)$ as a uniform spherical distribution in
$\mathcal{R}^d$: for any $|x| \le 1$, $f(x)$ is a constant; for $|x| > 1$, $f(x)
= 0$.  It is easy to see that $c_f$ in this situation is $2^d$, and thus for
some dataset $S$, $c$ must converge to that value as more and more points are
added to $S$.  Closed-form solutions for $c_f$ for more complex distributions
are less easy to derive; however, empirical speedup results from
\citet{langford2006} suggest the existence of datasets where $c$ is not strongly
dependent on $d$.  For instance, the \texttt{covtype} dataset has 54 dimensions
but the expansion constant is much smaller than other, lower-dimensional
datasets.



There are some other important observations about the behavior of $c$.  Adding a
single point to $S$ may increase $c$ arbitrarily: consider a set $S$ distributed
entirely on the surface of a unit hypersphere.  If one adds a single point at
the origin, producing the set $S'$, then $c$ explodes to $|S'|$ whereas before
it may have been much smaller than $|S|$.  Adding a single point may also
decrease $c$ significantly.  Suppose one adds a point arbitrarily close to the
origin to $S'$; now, the expansion constant will be $|S'| / 2$.  Both of these
situations are degenerate cases not commonly encountered in real-world behavior;
we discuss them in order to point out that although we can bound the behavior of
$c$ as $|S| \to \infty$ for $S$ from a stationary distribution, we are not able
to easily say much about its convergence behavior.

The expansion constant can be used to show a few useful bounds on various
properties of the cover tree; we restate these results below, given some cover
tree built on a dataset $S$ with expansion constant $c$ and $|S| = N$:

\begin{itemize}
  \item {\bf Width bound:} no cover tree node has more than $c^4$ children
(Lemma 4.1, \cite{langford2006}).

  \item {\bf Depth bound:} the maximum depth of any node is $O(c^2 \log N)$
(Lemma 4.3, \cite{langford2006}).

  \item {\bf Space bound:} a cover tree has $O(N)$ nodes (Theorem 1,
\cite{langford2006}).
\end{itemize}



Lastly, we introduce a convenience lemma of our own which is a generalization of
the packing arguments used by \citet{langford2006}.  This is a more flexible
version of their argument.

\begin{lemma}
Consider a dataset $S$ with expansion constant $c$ and a subset $C \subseteq S$
such that every point in $C$ is separated by $\delta$.  Then, for any point
$p$ (which may or may not be in $S$), and any radius $\rho \delta > 0$:
\begin{equation}
| B_S(p, \rho \delta) \cap C | \le c^{2 + \lceil \log_2 \rho \rceil}.
\end{equation}
\label{lem:packing}
\end{lemma}

\begin{proof}
The proof is based on the packing argument from Lemma 4.1 in
\cite{langford2006}. Consider two cases: first, let $d(p, p_i) > \rho \delta$
for any $p_i \in S$. In this case, $B_S(p, \rho \delta) = \emptyset$ and the
lemma holds trivially.
Otherwise, let $p_i \in S$ be a point such that $d(p, p_i) \leq \rho \delta$.
Observe that $B_S(p, \rho \delta) \subseteq B_S(p_i, 2 \rho \delta)$.
Also, $| B_S(p_i, 2 \rho \delta) | = c^{2 + \lceil \log_2 \rho
\rceil} | B_S(p_i, \delta / 2) |$ by the definition of the expansion constant.
Because each point in $C$ is separated by $\delta$, the
number of points in $B_S(p, \rho \delta) \cap C$ is
bounded by the number of disjoint balls of radius $\delta / 2$ that can be
packed into $B_S(p, \rho \delta)$.  In the worst case, this packing is
perfect, and

\begin{equation}
|B_S(p, \rho \delta)| \le \frac{|B_S(p_i, 2 \rho \delta)|}{|B_S(p_i, \delta
/ 2)|} \le c^{2 + \lceil \log_2 \rho \rceil}.
\end{equation}
\end{proof}

\section{Tree imbalance}
\label{sec:imbalance}

It is well-known that imbalance in trees leads to degradation in performance;
for instance, a $kd$-tree node with every descendant in its left child except
one is effectively useless.  A $kd$-tree full of nodes like this will perform
abysmally for nearest neighbor search, and it is not hard to generate a
pathological dataset that will cause a $kd$-tree of this sort.

This sort of imbalance applies to all types of trees, not just $kd$-trees.  In
our situation, we are interested in a better understanding of this imbalance for
cover trees, and thus endeavor to introduce a more formal measure of imbalance
which is correlated with tree performance.  Numerous measures of tree
imbalance have already been established; one example is that proposed by
\citet{colless1982review}, and another is Sackin's index \citep{sackin1972good},
but we aim to capture a different measure of imbalance that utilizes the leveled
structure of the cover tree.

We already know each node in a cover tree is indexed with an integer level (or
scale).  In the explicit representation of the cover tree, each non-leaf node
has children at a lower level.  But these children need not be strictly one
level lower; see Figure \ref{fig:imbalance}.  In Figure
\ref{fig:imbalance-good}, each cover tree node has children that are strictly
one level lower; we will refer to this as a {\em perfectly balanced cover tree}.
Figure \ref{fig:imbalance-bad}, on the other hand, contains the node
$\mathscr{N}_m$ which has two children with scale two less than $s_m$.  We will
refer to this as an {\em imbalanced cover tree}.  Note that in our definition,
the balance of a cover tree has nothing to do with differing number of
descendants in each child branch but instead only missing levels.

\begin{figure}
\begin{subfigure}[b]{0.585\textwidth}
  \begin{center}
    \begin{tikzpicture}[>=latex,line join=bevel,scale=0.47]
  \pgfsetlinewidth{1bp}
\begin{scope}
  \pgfsetstrokecolor{black}
  \definecolor{strokecol}{rgb}{1.0,1.0,1.0};
  \pgfsetstrokecolor{strokecol}
  \definecolor{fillcol}{rgb}{1.0,1.0,1.0};
  \pgfsetfillcolor{fillcol}
  \filldraw (0bp,0bp) -- (0bp,180bp) -- (486bp,180bp) -- (486bp,0bp) -- cycle;
\end{scope}
\begin{scope}
  \pgfsetstrokecolor{black}
  \definecolor{strokecol}{rgb}{1.0,1.0,1.0};
  \pgfsetstrokecolor{strokecol}
  \definecolor{fillcol}{rgb}{1.0,1.0,1.0};
  \pgfsetfillcolor{fillcol}
  \filldraw (0bp,0bp) -- (0bp,180bp) -- (486bp,180bp) -- (486bp,0bp) -- cycle;
\end{scope}
\begin{scope}
  \pgfsetstrokecolor{black}
  \definecolor{strokecol}{rgb}{1.0,1.0,1.0};
  \pgfsetstrokecolor{strokecol}
  \definecolor{fillcol}{rgb}{1.0,1.0,1.0};
  \pgfsetfillcolor{fillcol}
  \filldraw (0bp,0bp) -- (0bp,180bp) -- (486bp,180bp) -- (486bp,0bp) -- cycle;
\end{scope}
  \pgfsetcolor{black}
  \draw [-] (228.43bp,74.834bp) .. controls (218.25bp,64.938bp) and (204.48bp,51.546bp)  .. (185.8bp,33.385bp);
  \draw [-] (387bp,71.697bp) .. controls (387bp,63.983bp) and (387bp,54.712bp)  .. (387bp,36.104bp);
  \draw [-] (401.57bp,74.834bp) .. controls (411.75bp,64.938bp) and (425.52bp,51.546bp)  .. (444.2bp,33.385bp);
  \draw [-] (257.57bp,74.834bp) .. controls (267.75bp,64.938bp) and (281.52bp,51.546bp)  .. (300.2bp,33.385bp);
  \draw [-] (84.43bp,74.834bp) .. controls (74.25bp,64.938bp) and (60.476bp,51.546bp)  .. (41.796bp,33.385bp);
  \draw [-] (99bp,71.697bp) .. controls (99bp,63.983bp) and (99bp,54.712bp)  .. (99bp,36.104bp);
  \draw [-] (221.75bp,150.67bp) .. controls (197.4bp,138.83bp) and (157.28bp,119.33bp)  .. (120.33bp,101.37bp);
  \draw [-] (243bp,143.7bp) .. controls (243bp,135.98bp) and (243bp,126.71bp)  .. (243bp,108.1bp);
  \draw [-] (243bp,71.697bp) .. controls (243bp,63.983bp) and (243bp,54.712bp)  .. (243bp,36.104bp);
  \draw [-] (264.25bp,150.67bp) .. controls (288.6bp,138.83bp) and (328.72bp,119.33bp)  .. (365.67bp,101.37bp);
\begin{scope}
  \definecolor{strokecol}{rgb}{0.0,0.0,0.0};
  \pgfsetstrokecolor{strokecol}
  \definecolor{fillcol}{rgb}{0.83,0.83,0.83};
  \pgfsetfillcolor{fillcol}
  \filldraw [opacity=1] (243bp,18bp) ellipse (27bp and 18bp);
  \draw (243bp,18bp) node {$\mathscr{N}_h$};
\end{scope}
\begin{scope}
  \definecolor{strokecol}{rgb}{0.0,0.0,0.0};
  \pgfsetstrokecolor{strokecol}
  \definecolor{fillcol}{rgb}{0.83,0.83,0.83};
  \pgfsetfillcolor{fillcol}
  \filldraw [opacity=1] (387bp,18bp) ellipse (27bp and 18bp);
  \draw (387bp,18bp) node {$\mathscr{N}_j$};
\end{scope}
\begin{scope}
  \definecolor{strokecol}{rgb}{0.0,0.0,0.0};
  \pgfsetstrokecolor{strokecol}
  \definecolor{fillcol}{rgb}{0.83,0.83,0.83};
  \pgfsetfillcolor{fillcol}
  \filldraw [opacity=1] (315bp,18bp) ellipse (27bp and 18bp);
  \draw (315bp,18bp) node {$\mathscr{N}_i$};
\end{scope}
\begin{scope}
  \definecolor{strokecol}{rgb}{0.0,0.0,0.0};
  \pgfsetstrokecolor{strokecol}
  \definecolor{fillcol}{rgb}{0.83,0.83,0.83};
  \pgfsetfillcolor{fillcol}
  \filldraw [opacity=1] (387bp,90bp) ellipse (27bp and 18bp);
  \draw (387bp,90bp) node {$\mathscr{N}_d$};
\end{scope}
\begin{scope}
  \definecolor{strokecol}{rgb}{0.0,0.0,0.0};
  \pgfsetstrokecolor{strokecol}
  \definecolor{fillcol}{rgb}{0.83,0.83,0.83};
  \pgfsetfillcolor{fillcol}
  \filldraw [opacity=1] (243bp,90bp) ellipse (27bp and 18bp);
  \draw (243bp,90bp) node {$\mathscr{N}_c$};
\end{scope}
\begin{scope}
  \definecolor{strokecol}{rgb}{0.0,0.0,0.0};
  \pgfsetstrokecolor{strokecol}
  \definecolor{fillcol}{rgb}{0.83,0.83,0.83};
  \pgfsetfillcolor{fillcol}
  \filldraw [opacity=1] (99bp,90bp) ellipse (27bp and 18bp);
  \draw (99bp,90bp) node {$\mathscr{N}_b$};
\end{scope}
\begin{scope}
  \definecolor{strokecol}{rgb}{0.0,0.0,0.0};
  \pgfsetstrokecolor{strokecol}
  \definecolor{fillcol}{rgb}{0.83,0.83,0.83};
  \pgfsetfillcolor{fillcol}
  \filldraw [opacity=1] (459bp,18bp) ellipse (27bp and 18bp);
  \draw (459bp,18bp) node {$\mathscr{N}_k$};
\end{scope}
\begin{scope}
  \definecolor{strokecol}{rgb}{0.0,0.0,0.0};
  \pgfsetstrokecolor{strokecol}
  \definecolor{fillcol}{rgb}{0.83,0.83,0.83};
  \pgfsetfillcolor{fillcol}
  \filldraw [opacity=1] (243bp,162bp) ellipse (27bp and 18bp);
  \draw (243bp,162bp) node {$\mathscr{N}_a$};
\end{scope}
\begin{scope}
  \definecolor{strokecol}{rgb}{0.0,0.0,0.0};
  \pgfsetstrokecolor{strokecol}
  \definecolor{fillcol}{rgb}{0.83,0.83,0.83};
  \pgfsetfillcolor{fillcol}
  \filldraw [opacity=1] (171bp,18bp) ellipse (27bp and 18bp);
  \draw (171bp,18bp) node {$\mathscr{N}_g$};
\end{scope}
\begin{scope}
  \definecolor{strokecol}{rgb}{0.0,0.0,0.0};
  \pgfsetstrokecolor{strokecol}
  \definecolor{fillcol}{rgb}{0.83,0.83,0.83};
  \pgfsetfillcolor{fillcol}
  \filldraw [opacity=1] (99bp,18bp) ellipse (27bp and 18bp);
  \draw (99bp,18bp) node {$\mathscr{N}_f$};
\end{scope}
\begin{scope}
  \definecolor{strokecol}{rgb}{0.0,0.0,0.0};
  \pgfsetstrokecolor{strokecol}
  \definecolor{fillcol}{rgb}{0.83,0.83,0.83};
  \pgfsetfillcolor{fillcol}
  \filldraw [opacity=1] (27bp,18bp) ellipse (27bp and 18bp);
  \draw (27bp,18bp) node {$\mathscr{N}_e$};
\end{scope}
\draw[gray,thin,dashed] (0bp,58bp) -- (500bp,58bp);
\draw[gray,thin,dashed] (0bp,133bp) -- (500bp,133bp);

\draw (470bp,180bp) node[darkgray] {\scriptsize $s_a$};
\draw (470bp,122bp) node[darkgray] {\scriptsize $s_a - 1$};
\draw (470bp,47bp) node[darkgray] {\scriptsize $s_a - 2$};
\end{tikzpicture}
  \end{center}
  \caption{Balanced cover tree.}
  \label{fig:imbalance-good}
\end{subfigure}
\begin{subfigure}[b]{0.415\textwidth}
  \begin{center}
    \begin{tikzpicture}[>=latex,line join=bevel,scale=0.47]
  \pgfsetlinewidth{1bp}
\begin{scope}
  \pgfsetstrokecolor{black}
  \definecolor{strokecol}{rgb}{1.0,1.0,1.0};
  \pgfsetstrokecolor{strokecol}
  \definecolor{fillcol}{rgb}{1.0,1.0,1.0};
  \pgfsetfillcolor{fillcol}
  \filldraw (0bp,0bp) -- (0bp,180bp) -- (342bp,180bp) -- (342bp,0bp) -- cycle;
\end{scope}
  \pgfsetcolor{black}
  \draw [-] (156.43bp,74.834bp) .. controls (146.25bp,64.938bp) and (132.48bp,51.546bp)  .. (113.8bp,33.385bp);
  \draw [-] (156.4bp,146.6bp) .. controls (131.06bp,121.62bp) and (78.82bp,70.101bp)  .. (41.639bp,33.435bp);
  \draw [-] (185.57bp,74.834bp) .. controls (195.75bp,64.938bp) and (209.52bp,51.546bp)  .. (228.2bp,33.385bp);
  \draw [-] (185.6bp,146.6bp) .. controls (210.94bp,121.62bp) and (263.18bp,70.101bp)  .. (300.36bp,33.435bp);
  \draw [-] (171bp,143.7bp) .. controls (171bp,135.98bp) and (171bp,126.71bp)  .. (171bp,108.1bp);
  \draw [-] (171bp,71.697bp) .. controls (171bp,63.983bp) and (171bp,54.712bp)  .. (171bp,36.104bp);
\begin{scope}
  \definecolor{strokecol}{rgb}{0.0,0.0,0.0};
  \pgfsetstrokecolor{strokecol}
  \definecolor{fillcol}{rgb}{0.83,0.83,0.83};
  \pgfsetfillcolor{fillcol}
  \filldraw [opacity=1] (171bp,18bp) ellipse (27bp and 18bp);
  \draw (171bp,18bp) node {$\mathscr{N}_r$};
\end{scope}
\begin{scope}
  \definecolor{strokecol}{rgb}{0.0,0.0,0.0};
  \pgfsetstrokecolor{strokecol}
  \definecolor{fillcol}{rgb}{0.83,0.83,0.83};
  \pgfsetfillcolor{fillcol}
  \filldraw [opacity=1] (243bp,18bp) ellipse (27bp and 18bp);
  \draw (243bp,18bp) node {$\mathscr{N}_s$};
\end{scope}
\begin{scope}
  \definecolor{strokecol}{rgb}{0.0,0.0,0.0};
  \pgfsetstrokecolor{strokecol}
  \definecolor{fillcol}{rgb}{0.83,0.83,0.83};
  \pgfsetfillcolor{fillcol}
  \filldraw [opacity=1] (171bp,90bp) ellipse (27bp and 18bp);
  \draw (171bp,90bp) node {$\mathscr{N}_n$};
\end{scope}
\begin{scope}
  \definecolor{strokecol}{rgb}{0.0,0.0,0.0};
  \pgfsetstrokecolor{strokecol}
  \definecolor{fillcol}{rgb}{0.83,0.83,0.83};
  \pgfsetfillcolor{fillcol}
  \filldraw [opacity=1] (315bp,18bp) ellipse (27bp and 18bp);
  \draw (315bp,18bp) node {$\mathscr{N}_t$};
\end{scope}
\begin{scope}
  \definecolor{strokecol}{rgb}{0.0,0.0,0.0};
  \pgfsetstrokecolor{strokecol}
  \definecolor{fillcol}{rgb}{0.83,0.83,0.83};
  \pgfsetfillcolor{fillcol}
  \filldraw [opacity=1] (171bp,162bp) ellipse (27bp and 18bp);
  \draw (171bp,162bp) node {$\mathscr{N}_m$};
\end{scope}
\begin{scope}
  \definecolor{strokecol}{rgb}{0.0,0.0,0.0};
  \pgfsetstrokecolor{strokecol}
  \definecolor{fillcol}{rgb}{0.83,0.83,0.83};
  \pgfsetfillcolor{fillcol}
  \filldraw [opacity=1] (99bp,18bp) ellipse (27bp and 18bp);
  \draw (99bp,18bp) node {$\mathscr{N}_q$};
\end{scope}
\begin{scope}
  \definecolor{strokecol}{rgb}{0.0,0.0,0.0};
  \pgfsetstrokecolor{strokecol}
  \definecolor{fillcol}{rgb}{0.83,0.83,0.83};
  \pgfsetfillcolor{fillcol}
  \filldraw [opacity=1] (27bp,18bp) ellipse (27bp and 18bp);
  \draw (27bp,18bp) node {$\mathscr{N}_p$};
\end{scope}
\draw[darkgray, thin] (-30bp,0bp) -- (-30bp,190bp);
\draw[gray,thin,dashed] (0bp,58bp) -- (350bp,58bp);
\draw[gray,thin,dashed] (0bp,133bp) -- (350bp,133bp);

\draw (320bp,180bp) node[darkgray] {\scriptsize $s_m$};
\draw (320bp,122bp) node[darkgray] {\scriptsize $s_m - 1$};
\draw (320bp,47bp) node[darkgray] {\scriptsize $s_m - 2$};
\end{tikzpicture}
  \end{center}
  \caption{Imbalanced cover tree.}
  \label{fig:imbalance-bad}
\end{subfigure}
\caption{Balanced and imbalanced cover trees.}
\label{fig:imbalance}
\end{figure}
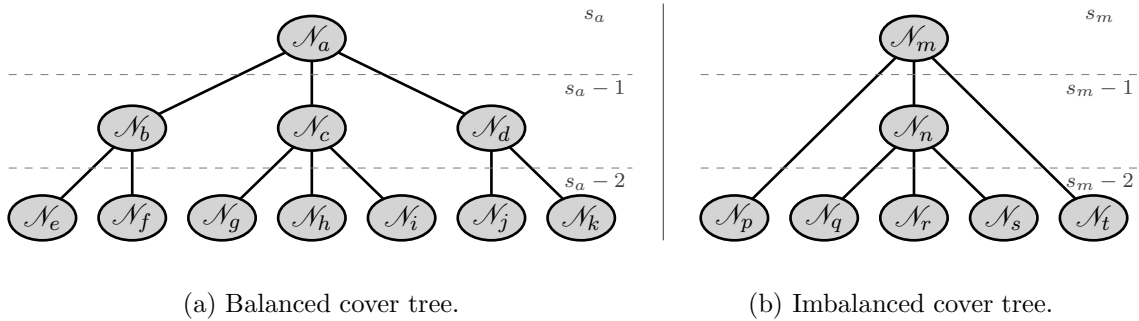

An imbalanced cover tree can happen in practice, and in the worst cases, the
imbalance may be far worse than the simple graphs of Figure \ref{fig:imbalance}.
Consider a dataset with a single outlier which is very far away from all of the
other points\footnote{Note also that for an outlier sufficiently far away, the
expansion constant is $N - 1$, so we should expect poor performance with the
cover tree anyway.}.  Figure \ref{fig:outlier} shows what happens in
this situation: the root node has two children; one of these children has only
the outlier as a descendant, and the other child has the rest of the points in
the dataset as a descendant.  In fact, it is easy to find datasets with a
handful of outliers that give rise to a chain-like structure at the top of the
tree: see Figure \ref{fig:outliers} for an illustration\footnote{As a side note,
this behavior is not limited to cover trees, and can happen to mean-split
$kd$-trees too, especially in higher dimensions.}.

\begin{figure}
\begin{center}
  \input{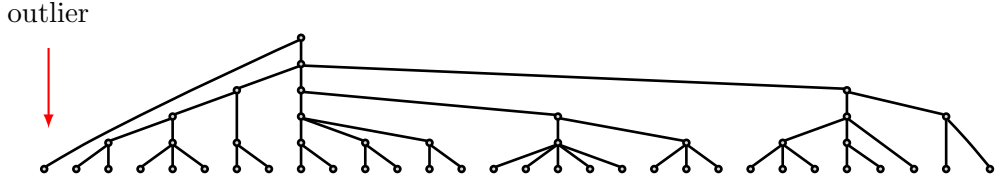}
\end{center}
\caption{Single-outlier cover tree.}
\label{fig:outlier}
\end{figure}

\begin{figure}
\begin{center}
  \input{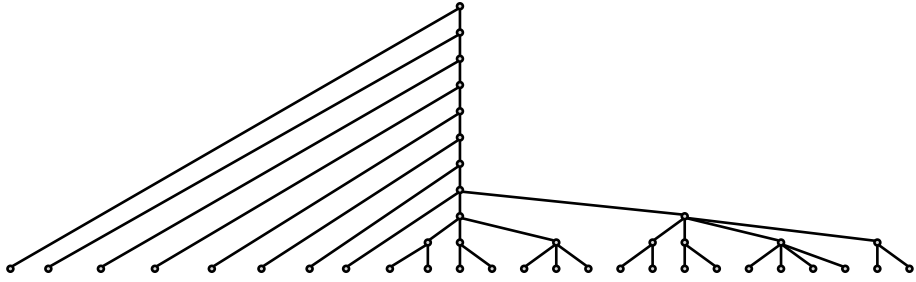}
\end{center}
\caption{A multiple-outlier cover tree.}
\label{fig:outliers}
\end{figure}

A tree that has this chain-like structure all the way down, which is similar to
the $kd$-tree example at the beginning of this section, is going to perform
horrendously; motivated by this observation, we define a measure of tree
imbalance.

\begin{defn}
The {\it cover node imbalance} $i_n(\mathscr{N}_i)$ for a cover tree node
$\mathscr{N}_i$ with scale $s_i$ in the cover tree $\mathscr{T}$ is defined as
the cumulative number of missing levels between the node and its parent
$\mathscr{N}_p$ (which has scale $s_p$).  If
the node is a leaf child (that is, $s_i = -\infty$), then number of missing
levels is defined as the difference between $s_p$ and $s_{\min} - 1$ where
$s_{\min}$ is the smallest scale of a non-leaf node in $\mathscr{T}$.  If
$\mathscr{N}_i$ is the root of the tree, then the cover node imbalance is 0.
Explicitly written, this calculation is

\begin{equation}
i_n(\mathscr{N}_i) = \begin{dcases*}
  s_p - s_i - 1 & if $\mathscr{N}_i$ is not a leaf and not the root node \\
  \max(s_p - s_{\min} - 1, \; 0) & if $\mathscr{N}_i$ is a leaf \\
  0 & if $\mathscr{N}_i$ is the root node.
  \end{dcases*}
  \label{eqn_node_imbalance}
\end{equation}
\end{defn}

This simple definition of cover node imbalance is easy to calculate, and using
it, we can generalize to a measure of imbalance for the full tree.

\begin{defn}
\label{def:imbalance}
The {\it cover tree imbalance} $i_t(\mathscr{T})$ for a cover tree $\mathscr{T}$
is defined as the cumulative number of missing levels in the tree.  This can be
expressed as a function of cover node imbalances easily:

\begin{equation}
i_t(\mathscr{T}) = \sum_{\mathscr{N}_i \in \mathscr{T}} i_n(\mathscr{N}_i).
\end{equation}
\end{defn}

A perfectly balanced cover tree $\mathscr{T}_b$ with no missing levels has
imbalance $i_t(\mathscr{T}_b) = 0$ (for instance, Figure
\ref{fig:imbalance-good}).  A worst-case cover tree $\mathscr{T}_w$ which is
entirely a chain-like structure with maximum scale $s_{\max}$ and minimum scale
$s_{\min}$ will have imbalance $i_t(\mathscr{T}_w) \sim N (s_{\max} -
s_{\min})$.  Because of this chain-like structure, each level has only one node
and thus there are at least $N$ levels; or, $s_{\max} - s_{\min} \ge N$, meaning
that in the worst case the imbalance is quadratic in $N$\footnote{Note that in
this situation, $c \sim N$ also.}.

However, for most real-world datasets with the cover tree implementation in {\bf
mlpack} \citep{mlpack2013} and the reference implementation
\citep{langford2006}, the tree imbalance is near-linear with the number of
points.  Generally, most of the cover tree imbalance is contributed by leaf
nodes whose parent has scale greater than $s_{\min}$.  At this time, no cover
tree construction algorithm specifically aims to minimize imbalance.




\section{General runtime bound}
\label{sec:bound}

Perhaps more interesting than measures of tree imbalance is the way cover trees
are actually used in dual-tree algorithms.  Although cover trees were originally
intended for nearest neighbor search \citep[See Algorithm
\texttt{Find-All-Nearest},][]{langford2006}, they can be adapted to a wide
variety of problems: minimum
spanning tree calculation \citep{march2010euclidean}, approximate nearest neighbor
search \citep{ram2009rank}, Gaussian processes posterior calculation
\citep{moore2014fast}, and max-kernel search \citep{curtin2014dual} are some
examples.  Further, through the tree-independent dual-tree algorithm abstraction
of \citet{curtin2013tree}, other existing dual-tree algorithms can easily be
adapted for use with cover trees.

In the framework of tree-independent dual-tree algorithms, all that is necessary
to describe a dual-tree algorithm is a point-to-point base case function
(\texttt{BaseCase()}) and a node-to-node pruning rule (\texttt{Score()}).  These
functions, which are often very straightforward, are then paired with a type of
tree and a pruning dual-tree traversal to produce a working algorithm.  In later
sections, we will consider specific examples.

\begin{algorithm}[tb]
  \begin{algorithmic}[1]
    \STATE {\bfseries Input:} query node $\mathscr{N}_q$, set of reference nodes
$R$ \label{alg:line:ct-dual-input}
    \STATE {\bfseries Output:} none

    \medskip
    \STATE $s^{\max}_r \gets \max_{\mathscr{N}_r \in R} s_r$
\label{alg:line:ct-dual-srmax}
    \IF{$(s_q < s^{\max}_r)$} \label{alg:line:ct-dual-ref-recursion-start}
      \STATE \COMMENT{Perform a reference recursion.}
      \FORALL{$\mathscr{N}_r \in R$} \label{alg:line:ct-dual-base-case-start}
        \STATE \texttt{BaseCase($p_q$, $p_r$)}
      \ENDFOR \label{alg:line:ct-dual-base-case-end}
      \STATE $R_r \gets \{ \mathscr{N}_r \in R : s_r = s^{\max}_r \}$
\label{alg:line:ct-dual-ref-set}
      \STATE $R_{r - 1} \gets \{ \mathscr{C}(\mathscr{N}_r) : \mathscr{N}_r \in
R_r \} \cup (R \setminus R_r)$ \label{alg:line:ct-dual-ref-children}
      \STATE $R'_{r - 1} \gets \{ \mathscr{N}_r \in R_{r - 1} :
\texttt{Score(}\mathscr{N}_q\texttt{,} \mathscr{N}_r\texttt{)} \ne \infty \}$
\label{alg:line:ct-dual-ref-score}
      \STATE recurse with $\mathscr{N}_q$ and $R'_{r - 1}$
\label{alg:line:ct-dual-ref-recursion-end}
    \ELSE \label{alg:line:ct-dual-query-recursion-start}
      \STATE \COMMENT{Perform a query recursion.}
      \FORALL{$\mathscr{N}_{qc} \in \mathscr{C}(\mathscr{N}_q)$}
        \STATE $R' \gets \{ \mathscr{N}_r \in R :
\texttt{Score(}\mathscr{N}_q\texttt{,} \mathscr{N}_r\texttt{)} \ne \infty \}$
\label{alg:line:ct-dual-query-pruning}
        \STATE recurse with $\mathscr{N}_{qc}$ and $R'$
\label{alg:line:ct-dual-query-recursion}
      \ENDFOR \label{alg:line:ct-dual-query-recursion-end}
    \ENDIF
  \end{algorithmic}
  \caption{The standard pruning dual-tree traversal for cover trees.}
  \label{alg:cover-tree-dual}
\end{algorithm}

When using cover trees, the typical pruning dual-tree traversal is an adapted
form of the original nearest neighbor search algorithm \citep[see
\texttt{Find-All-Nearest},][]{langford2006}; this traversal is implemented in
both the cover tree reference implementation and in the more flexible {\bf
mlpack} library \citep{mlpack2013}.  The problem-independent traversal is given
in Algorithm \ref{alg:cover-tree-dual} and was originally presented by
\citet{curtin2014dual}.  Initially, it is called with the root of the query tree
and a reference set $R$ containing only the root of the reference tree.

This dual-tree recursion is a depth-first recursion in the query tree and a
breadth-first recursion in the reference tree; to this end, the recursion
maintains one query node $\mathscr{N}_q$ and a reference set $R$.  The set $R$
may contain reference nodes with many different scales; the maximum scale in the
reference set is $s_r^{\max}$ (line \ref{alg:line:ct-dual-srmax}).  Each single
recursion will descend either the query tree or the reference tree, not both;
the conditional in line \ref{alg:line:ct-dual-ref-recursion-start}, which
determines whether the query or reference tree will be recursed, is aimed at
keeping the relative scales of query nodes and reference nodes close.

A query recursion (lines
\ref{alg:line:ct-dual-query-recursion-start}--\ref{alg:line:ct-dual-query-recursion-end})
is straightforward: for each child $\mathscr{N}_{qc}$ of $\mathscr{N}_q$, the
node combinations $(\mathscr{N}_{qc}, \mathscr{N}_r)$ are scored for each
$\mathscr{N}_r$ in the reference set $R$.  If possible, these combinations are
pruned to form the set $R'$ (line \ref{alg:line:ct-dual-query-recursion}) by
checking the output of the \texttt{Score()} function, and then the algorithm
recurses with $\mathscr{N}_{qc}$ and $R'$.

A reference recursion (lines
\ref{alg:line:ct-dual-ref-recursion-start}--\ref{alg:line:ct-dual-ref-recursion-end})
is similar to a query recursion, but the pruning strategy is significantly more
complicated.  Given $R$, we calculate $R_r$, which is the set of nodes in $R$
that have scale $s_r^{\max}$.  Then, for each node $\mathscr{N}_r$ in the set of
children of nodes in $R_r$, the node combinations $(\mathscr{N}_q,
\mathscr{N}_r)$ are scored and pruned if possible.  Those reference nodes that
were not pruned form the set $R'_{r - 1}$.  Then, this set is combined with $R
\setminus R_r$---that is, each of the nodes in $R$ that was {\em not} recursed
into---to produce $R'$, and the algorithm recurses with $\mathscr{N}_q$ and the
reference set $R'$.

The reference recursion only recurses into the top-level subset of the reference
nodes in order to preserve the separation invariant.  It is easy to show that
every pair of points held in nodes in $R$ is separated by at least
$2^{s_r^{\max}}$:

\begin{lemma}
For all nodes $\mathscr{N}_i, \mathscr{N}_j \in R$ (in the context of Algorithm
\ref{alg:cover-tree-dual}) which contain points $p_i$ and $p_j$, respectively,
$d(p_i, p_j) > 2^{s_r^{\max}}$, with $s_r^{\max}$ defined as in line
\ref{alg:line:ct-dual-srmax}.
\end{lemma}

\begin{proof}
This proof is by induction.  If $|R| = 1$, such as during the first reference
recursion, the result obviously holds.  Now consider any reference set $R$ and
assume the statement of the lemma holds for this set $R$, and define
$s_r^{\max}$ as the maximum scale of any node in $R$.  Construct the set
$R_{r - 1}$ as in line \ref{alg:line:ct-dual-ref-children} of Algorithm
\ref{alg:cover-tree-dual}; if $| R_{r - 1} | \le 1$, then $R_{r - 1}$ satisfies
the desired property.

Otherwise, take any $\mathscr{N}_i, \mathscr{N}_j$ in $R_{r - 1}$, with points
$p_i$ and $p_j$, respectively, and scales $s_i$ and $s_j$, respectively.
Clearly, if $s_i = s_j = s_r^{\max} - 1$, then by the separation invariant
$d(p_i, p_j) > 2^{s_r^{\max} - 1}$.

Now suppose that $s_i < s_r^{\max} - 1$.  This implies that there exists some
implicit cover tree node with point $p_i$ and scale $s_r^{\max} - 1$ (as well an
implicit child of this node $p_i$ with scale $s_r^{\max} - 2$ and so forth until
one of these implicit nodes has child $p_i$ with scale $s_i$).  Because the
separation invariant applies to both implicit and explicit representations of
the tree, we conclude that $d(p_i, p_r) > 2^{s_r^{\max}} - 1$.  The same
argument may be made for the case where $s_j < s_r^{\max} - 1$, with the same
conclusion.

We may therefore conclude that each point of each node in $R_{r - 1}$ is
separated by $2^{s_r^{\max} - 1}$.  Note that $R'_{r - 1} \subseteq R_{r - 1}$
and that $R \setminus R_{r - 1} \subseteq R$ in order to see that this condition
holds for all nodes in $R' = R'_{r - 1} \cup (R \setminus R_{r - 1})$.

Because we have shown that the condition holds for the initial reference set and
for any reference set produced by a reference recursion, we have shown that the
statement of the lemma is true.
\end{proof}

This observation means that the set of points $P$ held by all nodes in $R$ is
always a subset of $C_{s_r^{\max}}$.  This fact will be useful in our later
runtime proofs.

Next, we develop notions with which to understand the behavior of the cover tree
dual-tree traversal when the datasets are of significantly different scale
distributions.

If the datasets are similar in scale distribution (that is, inter-point distances
tend to follow the same distribution), then the recursion will alternate between
query recursions and reference recursions.  But if the query set contains points
which are, in general, much farther apart than the reference set, then the
recursion will start with many query recursions before reaching a reference
recursion.  The converse case also holds.  We are interested in formalizing
this notion of scale distribution; therefore, define the following
dataset-dependent constants for the query set $S_q$ and the reference set $S_r$:

\begin{itemize}
  \item $\eta_q$: the largest pairwise distance in $S_q$
  \item $\delta_q$: the smallest nonzero pairwise distance in $S_q$
  \item $\eta_r$: the largest pairwise distance in $S_r$
  \item $\delta_r$: the smallest nonzero pairwise distance in $S_r$
\end{itemize}

These constants are directly related to the aspect ratio of the datasets;
indeed, $\eta_q / \delta_q$ is exactly the aspect ratio of $S_q$.  Further,
let us define and bound the top and bottom levels of each tree:

\begin{itemize}
  \item The {\it top scale} $s_q^T$ of the query tree $\mathscr{T}_q$ is such
that as $\lceil \log_2(\eta_q) \rceil - 1 \le s_q^T \le \lceil \log_2(\eta_q)
\rceil$.
  \item The {\it minimum scale} of the query tree $\mathscr{T}_r$ is defined as
$s_q^{\min} = \lceil \log_2(\delta_q) \rceil$.
  \item The top scale $s_r^T$ of the reference tree $\mathscr{T}_r$ is such that
as $\lceil \log_2(\eta_r) \rceil - 1 \le s_r^T \le \lceil \log_2(\eta_r)
\rceil$.
  \item The minimum scale of the reference tree $\mathscr{T}_r$ is defined as
$s_r^{\min} = \lceil \log_2(\delta_r) \rceil$.
\end{itemize}

Note that the minimum scale is not the minimum scale of {\it any} cover tree
node (that would be $-\infty$), but the minimum scale of any non-leaf node in
the tree.

Suppose that our datasets are of a similar scale distribution: $s_q^T = s_r^T$,
and $s_q^{\min} = s_r^{\min}$.  In this setting we will have alternating query
and reference recursions.  But if this is not the case, then we have extra
reference recursions before the first query recursion or after the last query
recursion (situations where both these cases happen are possible).  Motivated by
this observation, let us quantify these extra reference recursions:

\begin{lemma}
\label{lem:extcase1}
For a dual-tree algorithm with $S_q \sim S_r \sim O(N)$ using cover trees and
the traversal given in Algorithm \ref{alg:cover-tree-dual}, the number of extra
reference recursions that happen before the first query recursion is bounded by

\begin{equation}
\min\left(O(N), \log_2(\eta_r / \eta_q) - 1\right).
\end{equation}
\end{lemma}

\begin{proof}
The first query recursion happens once $s_q \ge s_r^{\max}$.  The number of
reference recursions before the first query recursion is then bounded as the
number of levels in the reference tree between $s_r^T$ and $s_q^T$ that have at
least one explicit node.  Because there are $O(N)$ nodes in the reference tree,
the number of levels cannot be greater than $O(N)$ and thus the result holds.

The second bound holds by applying the definitions of $s_r^T$ and $s_q^T$ to the
expression $s_r^T - s_q^T - 1$:

\begin{eqnarray}
s_r^T - s_q^T - 1 &\le& \lceil \log_2(\eta_r) \rceil - (\lceil \log_2(\eta_q)
\rceil - 1) - 1 \\
&\le& \log_2(\eta_r) + 1 - \log_2(\eta_q)
\end{eqnarray}

\noindent which gives the statement of the lemma after applying logarithmic
identities.
\end{proof}

Note that the $O(N)$ bound may be somewhat loose, but it suffices for our later
purposes.  Now let us consider the other case:

\begin{lemma}
\label{lem:extcase3}
For a dual-tree algorithm with $S_q \sim S_r \sim O(N)$ using cover trees and
the traversal given in Algorithm \ref{alg:cover-tree-dual}, the number of extra
reference recursions that happen after the last query recursion is bounded by

\begin{equation}
\theta = \max\left\{\min\left[O(N \log_2(\delta_q / \delta_r)),
O(N^2)\right], \; 0\right\}.
\end{equation}
\end{lemma}

\begin{proof}
Our goal here is to count the number of reference recursions after the final
query recursion at level $s_q^{\min}$; the first of these reference recursions
is at scale $s_r^{\max} = s_q^{\min}$.  Because query nodes are not pruned in
this traversal, each reference recursion we are counting will be duplicated over
the whole set of $O(N)$ query nodes.  The first part of the bound follows by
observing that
$s_q^{\min} - s_r^{\min} \le \lceil \log_2(\delta_q) \rceil -
\lceil \log_2(\delta_r) \rceil - 1 \le \log_2(\delta_q / \delta_r)$.

The second part follows by simply observing that there are $O(N)$ reference
nodes.
\end{proof}

These two previous lemmas allow us a better understanding of what happens as the
reference set and query set become different.  Lemma \ref{lem:extcase1} shows
that the number of extra recursions caused by a reference set with larger
pairwise distances than the query set ($\eta_r$ larger than $\eta_q$) is modest;
on the other hand, Lemma \ref{lem:extcase3} shows that for each extra level in
the reference tree below $s_q^{\min}$, $O(N)$ extra recursions are required.
Using these lemmas and this intuition, we will prove general runtime bounds for
the cover tree traversal.

\begin{thm}
\label{thm:ct-runtime}
Given a reference set $S_r$ of size $N$ with an expansion constant $c_r$ and a
set of queries $S_q$ of size $O(N)$, a standard cover tree based dual-tree
algorithm (Algorithm \ref{alg:cover-tree-dual}) takes

\begin{equation}
O\left(c_r^4 | R^* | \chi \psi (N + i_t(\mathscr{T}_q) + \theta)\right)
\end{equation}

\noindent time, where $ | R^* | $ is the maximum size of the reference set $R$
(line \ref{alg:line:ct-dual-input}) during the dual-tree recursion, $\chi$ is
the maximum possible runtime of \texttt{BaseCase()}, $\psi$ is the maximum
possible runtime of \texttt{Score()}, and $\theta$ is defined as in Lemma
\ref{lem:extcase3}.
\end{thm}

\begin{proof}
First, split the algorithm into two parts: reference recursions (lines
\ref{alg:line:ct-dual-ref-recursion-start}--\ref{alg:line:ct-dual-ref-recursion-end})
and query recursions (lines
\ref{alg:line:ct-dual-query-recursion-start}--\ref{alg:line:ct-dual-query-recursion-end}).
The runtime of the algorithm is bounded as the runtime of a reference recursion
times the total number of reference recursions plus the total runtime of all
query recursions.

Consider a reference recursion (lines
\ref{alg:line:ct-dual-ref-recursion-start}--\ref{alg:line:ct-dual-ref-recursion-end}).
Define $R^*$ to be the largest set $R$ for any scale
$s_r^{\max}$ and any query node $\mathscr{N}_q$ during the course of the
algorithm; then, it is true that $| R | \le | R^* |$.  The work done in the base
case loop from lines
\ref{alg:line:ct-dual-base-case-start}--\ref{alg:line:ct-dual-base-case-end} is
thus $O(\chi | R |) \le O(\chi | R^*|)$.
Then, lines \ref{alg:line:ct-dual-ref-children} and
\ref{alg:line:ct-dual-ref-score} take $O(c_r^4 \psi | R |) \le O(c_r^4 \psi |
R^* |)$ time, because each reference node has up to $c_r^4$ children.  So, one
full reference recursion takes $O(c_r^4 \psi \chi | R^* |)$ time.

Now, note that there are $O(N)$ nodes in $\mathscr{T}_q$.  Thus, line
\ref{alg:line:ct-dual-query-recursion} is visited $O(N)$ times.  The amount of
work in line \ref{alg:line:ct-dual-query-pruning}, like in the reference
recursion, is bounded as $O(c_r^4 \psi | R^* |)$. Therefore, the total runtime
of all query recursions is $O(c_r^4 \psi | R^* | N)$.

Lastly, we must bound the total number of reference recursions.  Reference
recursions happen in three cases: \textit{(1)} $s_r^{\max}$ is greater than the
scale of the root of the query tree (no query recursions have happened yet);
\textit{(2)} $s_r^{\max}$ is less than or equal to the scale of the root of the
query tree, but is greater than the minimum scale of the query tree that is not
$-\infty$; \textit{(3)} $s_r^{\max}$ is less than the minimum scale of the query
tree that is not $-\infty$.

First, consider case \textit{(1)}.  Lemma \ref{lem:extcase1} shows that the
number of reference recursions of this type is bounded by $O(N)$.  Although
there is also a bound that depends on the sizes of the datasets, we only aim to
show a linear runtime bound, so the $O(N)$ bound is sufficient here.

Next, consider case \textit{(2)}.  In this situation, each query recursion
implies at least one reference recursion before another query recursion.  For
some query node $\mathscr{N}_q$, the exact number of reference recursions before
the children of $\mathscr{N}_q$ are recursed into is bounded above by
$i_n(\mathscr{N}_q) + 1$: if $\mathscr{N}_q$ has imbalance $0$, then it is
exactly one level below its parent, and thus there is only one reference
recursion.  On the other hand, if $\mathscr{N}_q$ is many levels below its
parent, then it is possible that a reference recursion may occur for each level
in between; this is a maximum of $i_n(\mathscr{N}_q) + 1$.

Because each query node in $\mathscr{T}_q$ is recursed into once, the total
number of reference recursions before each query recursion is

\begin{equation}
\sum_{\mathscr{N}_q \in \mathscr{T}_q} i_n(\mathscr{N}_q) + 1 = i_t(\mathscr{T}_q) +
O(N)
\end{equation}

\noindent since there are $O(N)$ nodes in the query tree.

Lastly, for case \textit{(3)}, we may refer to Lemma \ref{lem:extcase3}, giving
a bound of $\theta$ reference recursions in this case.

We may now combine these results for the runtime of a query recursions with the
total number of reference recursions in order to give the result of the theorem:

\begin{equation}
O\left(c_r^4 |R^*| \psi \chi \left(N + i_t(\mathscr{T}_q) + \theta\right)\right)
+ O\left(c_r^4 |R^*| \psi N\right) \sim O\left(c_r^4 |R^*| \psi
\chi \left(N + i_t(\mathscr{T}_q) + \theta\right)\right).
\end{equation}
\end{proof}

When we consider the monochromatic case (where $S_q = S_r$), the results
trivially simplify.

\begin{cor}
\label{cor:ct-runtime-mono}
Given the situation of Theorem \ref{thm:ct-runtime} but with $S_q = S_r = S$ so
that $c_q = c_r = c$ and $\mathscr{T}_q = \mathscr{T}_r = \mathscr{T}$, a
dual-tree algorithm using the standard cover tree traversal (Algorithm
\ref{alg:cover-tree-dual}) takes

\begin{equation}
O\left(c^4 |R^*| \chi \psi \left(N + i_t(\mathscr{T})\right)\right)
\end{equation}

\noindent time, where $ | R^* | $ is the maximum size of the reference set $R$
(line \ref{alg:line:ct-dual-input}) during the dual-tree recursion, $\chi$ is
the maximum possible runtime of \texttt{BaseCase()}, and $\psi$ is the maximum
possible runtime of \texttt{Score()}.
\end{cor}

An intuitive understanding of these bounds is best achieved by first considering
the monochromatic case (this case arises, for instance, in all-nearest-neighbor
search).  The linear dependence on $N$ arises from the fact that all query nodes
must be visited.  The dependence on the reference tree, however, is encapsulated
by the term $c^4 |R^*|$, with $|R^*|$ being the maximum size of the reference
set $R$; this value must be derived for each specific problem.  The bad
performance of poorly-behaved datasets with large $c$ (or, in the worst case, $c
\sim N$) is then captured in both of those terms.  Poorly-behaved datasets may
also have a high cover tree imbalance $i_t(\mathscr{T})$; the linear dependence of
runtime on imbalance is thus sensible for well-behaved datasets.

The bichromatic case ($S_q \ne S_r$) is a slightly more complex result which
deserves a bit more attention.  The intuition for all terms except $\theta$
remain virtually the same.

The term $\theta$ captures the effect of query and reference datasets with
different widths, and has one unfortunate corner case: when $\delta_q > \eta_r$,
then the query tree must be entirely descended before any reference recursion.
This results in a bound of the form $O(N \log (\eta_r / \delta_r))$, or
$O(N^2)$ (See Lemma \ref{lem:extcase3}).  This is because the reference tree
must be descended individually for each query point.

The quantity $|R^*|$ bounds the amount of work that needs to be done for each
recursion. In the worst case, $|R^*|$ can be $N$. However,
dual-tree algorithms rely on branch-and-bound techniques to prune away
work (Lines \ref{alg:line:ct-dual-ref-score} and
\ref{alg:line:ct-dual-query-pruning} in Algorithm \ref{alg:cover-tree-dual}). A
small value of $|R^*|$ will imply that the algorithm is extremely successful in
pruning away work.  An (upper) bound on $|R^*|$ (and the algorithm's
success in pruning work) will depend on the problem and the data.  As we will
show, bounding $|R^*|$ is often possible. 
For many dual-tree algorithms, $\chi \sim \psi \sim O(1)$; often, cached
sufficient statistics \citep{moore2000anchors} can enable $O(1)$ runtime
implementations of \texttt{BaseCase()} and \texttt{Score()}.

These results hold for any dual-tree algorithm regardless of the problem. Hence,
the runtime of any dual-tree algorithm
would be at least $O(N)$ using our bound, which matches the intuition that
answering $O(N)$ queries would take at least $O(N)$ time. For a particular
problem and data, if $c_r$, $|R^*|$, $\chi$, $\psi$ are bounded by constants
independent of $N$ and $\theta$ is no more than linear in $N$ (for large enough
$N$), then the dual-tree algorithm for that problem has a runtime linear in $N$.
Our theoretical result separates out the problem-dependent and the
problem-independent elements of the runtime bound, which allows us to simply
plug in the problem-dependent bounds in order to get runtime bounds for any
dual-tree algorithm without requiring an analysis from scratch.

Our results are similar to that of \citet{ram2009}, but those results depend on
a quantity called the {\it constant of bichromaticity}, denoted $\kappa$, which
has unclear relation to cover tree imbalance.  The dependence on $\kappa$ is
given as $c_q^{4 \kappa}$, which is not a good bound, especially because
$\kappa$ may be much greater than $1$ in the bichromatic case (where $S_q \ne
S_r$).

The more recent results of \citet{curtin2014dual} are more related to these
results, but they depend on the {\it inverse constant of bichromaticity} $\nu$
which suffers from the same problem as $\kappa$.  Although the dependence on
$\nu$ is linear (that is, $O(\nu N)$), bounding $\nu$ is difficult and it is not
true that $\nu = 1$ in the monochromatic case.

$\nu$ corresponds to the maximum number of reference recursions between a single
query recursion, and $\kappa$ corresponds to the maximum number of query
recursions between a single reference recursion.  The respective proofs that use
these constants then apply them as a worst-case measure for the whole algorithm:
when using $\kappa$, \citet{ram2009} assume that {\it every} reference recursion
may be followed by $\kappa$ query recursions; similarly, \citet{curtin2014dual}
assume that {\it every} query recursion may be followed by $\nu$ reference
recursions.  Here, we have simply used $i_t(\mathscr{T}_q)$ and $\theta$ as an
exact summation of the total extra reference recursions, which gives us a much
tighter bound than $\nu$ or $\kappa$ on the running time of the whole algorithm.

Further, both $\nu$ and $\kappa$ are difficult to empirically calculate and
require an entire run of the dual-tree algorithm.  On the other hand, bounding
$i_t(\mathscr{T}_q)$ (and $\theta$) can be done in one pass of the tree
(assuming the tree is already built).  Thus, not only is our bound tighter when
the cover tree imbalance is sublinear in $N$, it more closely reflects the
actual behavior of dual-tree algorithms, and the constants which it depends upon
are straightforward to calculate.

In the following sections, we will apply our results to specific problems and
show the utility of our bound in simplifying runtime proofs for dual-tree
algorithms.

\section{Nearest neighbor search}
\label{sec:nns}

\begin{algorithm}[tb]
\begin{algorithmic}
    \STATE {\bfseries Input:} query point $p_q$, reference point $p_r$, list of
candidate neighbors $N$ and distances $D$
    \STATE {\bfseries Output:} distance $d$ between $p_q$ and $p_r$

    \medskip

    \IF{$d(p_q, p_r) < D[p_q]$ \AND \texttt{BaseCase($p_q$, $p_r$)} not yet called}
    \STATE  $D[p_q] \gets d(p_q, p_r)$, and $N[p_q] \gets p_r$
    \ENDIF

    \RETURN $d(p_q, p_r)$
  \end{algorithmic}

  \caption{Nearest neighbor search \texttt{BaseCase()}}
  \label{alg:nn_base_case}
\end{algorithm}

\begin{algorithm}[tb]
  \begin{algorithmic}
    \STATE {\bfseries Input:} query node $\mathscr{N}_q$, reference node
$\mathscr{N}_r$
    \STATE {\bfseries Output:} a score for the node combination $(\mathscr{N}_q,
\mathscr{N}_r)$, or $\infty$ if the combination should be pruned

    \medskip

    \IF{$d_{\min}(\mathscr{N}_q, \mathscr{N}_r) < B(\mathscr{N}_q)$}
      \RETURN $d_{\min}(\mathscr{N}_q, \mathscr{N}_r)$
    \ENDIF

    \RETURN $\infty$
  \end{algorithmic}

  \caption{Nearest neighbor search \texttt{Score()}}
  \label{alg:nn_score}
\end{algorithm}

The standard task of nearest neighbor search can be simply described: given a
query set $S_q$ and a reference set $S_r$, for each query point $p_q \in S_q$,
find the nearest neighbor $p_r$ in the reference set $S_r$.  The task is
well-studied and well-known, and there exist numerous approaches for both exact
and approximate nearest neighbor search, including the cover tree nearest
neighbor search algorithm due to \citet{langford2006}.  We will consider that
algorithm, but in a tree-independent sense as given by \citet{curtin2013tree};
this means that to describe the algorithm, we require only a \texttt{BaseCase()}
and \texttt{Score()} function; these are given in Algorithms
\ref{alg:nn_base_case} and \ref{alg:nn_score}, respectively.  The
point-to-point \texttt{BaseCase()} function compares a query point $p_q$ and a
reference point $p_r$, updating the list of candidate neighbors for $p_q$ if
necessary.

The node-to-node \texttt{Score()} function determines if the entire
subtree of nodes under the reference node $\mathscr{N}_r$ can improve the
candidate neighbors for all descendant points of the query node $\mathscr{N}_q$;
if not, the node combination is pruned.  The \texttt{Score()} function depends
on the function $d_{\min}(\cdot, \cdot)$, which represents the minimum possible
distance between any two descendants of two nodes.  Its definition for cover
tree nodes is

\begin{equation}
d_{\min}(\mathscr{N}_q, \mathscr{N}_r) = d(p_q, p_r) - 2^{s_q + 1} - 2^{s_r +
1}.
\end{equation}

Given a type of tree and traversal, these two functions store the current
nearest neighbor candidates in the array $N$ and their distances in the array
$D$. \citep[See][for a more complete
discussion of how this algorithm works and a proof of
correctness.]{curtin2013tree}  The
\texttt{Score()} function depends on a bound function $B(\mathscr{N}_q)$ which
represents the maximum distance that could possibly improve a nearest neighbor
candidate for any descendant point of the query node $\mathscr{N}_q$.  The
standard bound function $B(\mathscr{N}_q)$ used for cover trees is adapted from
\citep{langford2006}:

\begin{equation}
B(\mathscr{N}_q) := D[p_q] + 2^{s_q + 1}
\end{equation}





In this formulation, the query node $\mathscr{N}_q$ holds the the query point
$p_q$, the quantity $D[p_q]$ is the current nearest neighbor candidate distance
for the query point $p_q$, and $2^{s_q + 1}$ corresponds to the furthest
descendant distance of $\mathscr{N}_q$.  For notational convenience in the
following proof, take $c_{qr} = \max((\max_{p_q \in S_q} c'_r), c_r)$, where
$c'_r$ is the expansion constant of the set $S_r \cup \{ p_q \}$.



\begin{thm}
Using cover trees, the standard cover tree pruning dual-tree traversal, and the
nearest neighbor search \texttt{BaseCase()} and \texttt{Score()} as given in
Algorithms \ref{alg:nn_base_case} and \ref{alg:nn_score}, respectively, and also
given a reference set $S_r$ with expansion constant $c_r$, and a query set $S_q$,
the running time of the algorithm is bounded by $O(c_r^4 c_{qr}^5 (N +
i_t(\mathscr{T}_q) + \theta))$ with $i_t(\mathscr{T}_q$ and $\theta$ defined as
in Definition \ref{def:imbalance} and Lemma \ref{lem:extcase3}, respectively.
\label{thm:nns}
\end{thm}

\begin{proof}
The running time of \texttt{BaseCase()} and \texttt{Score()} are clearly $O(1)$.
Due to Theorem \ref{thm:ct-runtime}, we therefore know that the runtime of the
algorithm is bounded by $O(c_r^4 |R^*| (N + i_t(\mathscr{T}_q) + \theta))$.
Thus, the only thing that remains is to bound the maximum size of the reference
set, $|R^*|$.

Assume that when $R^*$ is encountered, the maximum reference scale is
$s_r^{\max}$ and the query node is $\mathscr{N}_q$.  Every node $\mathscr{N}_r
\in R^*$ satisfies the property enforced in line
\ref{alg:line:ct-dual-ref-score} that
$d_{\min}(\mathscr{N}_q, \mathscr{N}_r) \le B(\mathscr{N}_q)$.
Using the definition of $d_{\min}(\cdot, \cdot)$ and $B(\cdot)$, we
expand the equation.  Note that $p_q$ is the point held in $\mathscr{N}_q$ and
$p_r$ is the point held in $\mathscr{N}_r$.  Also, take $\hat{p}_r$ to be the
current nearest neighbor candidate for $p_q$; that is, $D[p_q] = d(p_q,
\hat{p}_r)$ and $N[p_q] = \hat{p}_r$.  Then,

\begin{eqnarray}
d_{\min}(\mathscr{N}_q, \mathscr{N}_r) &\le& B(\mathscr{N}_q) \\
d(p_q, p_r) &\le& d(p_q, \hat{p}_r) + 2^{s_q + 1} + 2^{s_r + 1} + 2^{s_q + 1}
\label{eqn:pr_dist} \\
 &\le& d(p_q, \hat{p}_r) + 2(2^{s_r^{\max} + 1})
\end{eqnarray}

\noindent where the last step follows because $s_q + 1 \le s_r^{\max}$ and $s_r
\le s_r^{\max}$.  Define the set of points $P$ as the points held in each node
in $R^*$ (that is, $P = \{ p_r \in \mathscr{P}(\mathscr{N}_r) : \mathscr{N}_r
\in R^* \}$).  Then, we can write

%

\begin{equation}
P \subseteq B_{S_r}(p_q, d(p_q, \hat{p}_r) + 2(2^{s_r^{\max} + 1})).
\end{equation}



Suppose that the true nearest neighbor is $p_r^*$ and $d(p_q, p_r^*) >
2^{s_r^{\max} + 1}$.  Then, $p_r^*$ must be held as a descendant point of some
node in $R^*$ which holds some point $\tilde{p}_r$.  Using the triangle
inequality,

\begin{equation}
d(p_q, \hat{p}_r) \le d(p_q, \tilde{p}_r)
 \le d(p_q, p_r^*) + d(\tilde{p}_r, p_r^*)
 \le d(p_q, p_r^*) + 2^{s_r^{\max} + 1}.
\end{equation}

This gives that
$P \subseteq B_{S_r \cup \{ p_q \}}(p_q, d(p_q, p_r^*) + 3(2^{s_r^{\max} +
1}))$.
The previous step is necessary: to apply the definition of the expansion
constant, the ball must be centered at a point in the set; now, the center
($p_q$) is part of the set.

\begin{eqnarray}
| B_{S_r \cup \{ p_q \}}(p_q, d(p_q, p_r^*) + 3(2^{s_r^{\max} + 1})) | &\le& |
B_{S_r \cup \{ p_q \}}(p_q, 4 d(p_q, p_r^*)) | \\
 &\le& c_{qr}^3 | B_{S_r \cup \{ p_q \}}(p_q, d(p_q, p_r^*) / 2) |
\end{eqnarray}

\noindent which follows because the expansion constant of the set $S_r \cup \{
p_q \}$ is bounded above by $c_{qr}$.  Next, we know that $p_r^*$ is the closest
point to $p_q$ in $S_r \cup \{ p_q \}$; thus, there cannot exist a point $p'_r
\ne p_q \in S_r \cup \{ p_q \}$ such that
$p'_r \in B_{S_{qr}}(p_q, d(p_q, p_r^*) / 2)$
because that would imply that $d(p_q, p'_r) < d(p_q, p_r^*)$, which is
a contradiction.  Thus, the only point in the ball is $p_q$, and we have that $|
B_{S_r \cup \{ p_q \}}(p_q, d(p_q, p_r^*) / 2) | = 1$, giving the result that
$|R| \le c_{qr}^3$ in this case.

The other case is when $d(p_q, p_r^*) \le 2^{s_r^{\max} + 1}$, which means that
$d(p_q, \hat{p}_r) \le 2^{s_r^{\max} + 2}$. 
Note that $P \in C_{s_r^{\max}}$, and therefore

\begin{eqnarray}
P &\subseteq& B_{S_r}(p_q, d(p_q, p_r^*) + 3(2^{s_r^{\max} + 1})) \cap
C_{s_r^{\max}} \\
 &\subseteq& B_{S_r}(p_q, 4(2^{s_r^{\max} + 1})) \cap C_{s_r^{\max}}.
\end{eqnarray}

Every point in $C_{s_r^{\max}}$ is separated by at least $2^{s_r^{\max}}$.
Using Lemma \ref{lem:packing} with $\delta = 2^{s_r^{\max}}$ and $\rho = 8$
yields that $|P| \le c_r^5$.  This gives the result, because $c_r^5 \le
c_{qr}^5$.
\end{proof}




%


In the monochromatic case where $S_q = S_r$, the bound is $O(c^9 (N +
i_t(\mathscr{T}))$ because $c = c_r = c_{qr}$ and $\theta = 0$.  For
well-behaved trees where $i_t(\mathscr{T}_q)$ is linear or sublinear in $N$,
this represents the current tightest worst-case runtime bound for nearest
neighbor search.

\section{Approximate kernel density estimation}
\label{sec:akde}

\citet{ram2009} present a clever technique for bounding the
running time of approximate kernel density estimation based on the properties of
the kernel, when the kernel is shift-invariant and satisfies a few assumptions.
We will restate these assumptions and provide an adapted proof using Theorem
\ref{thm:ct-runtime}, which gives a tighter bound.

Approximate kernel density estimation is a common application of dual-tree
algorithms \citep{gray2003nonparametric, nbody}.  Given a query set $S_q$, a
reference set $S_r$ of size $N$, and a kernel function $\mathcal{K}(\cdot,
\cdot)$, the true kernel density estimate for a query point $p_q$ is given as

\begin{equation}
f^*(p_q) = \sum_{p_r \in S_r} \mathcal{K}(p_q, p_r).
\end{equation}

In the case of an infinite-tailed kernel $\mathcal{K}(\cdot, \cdot)$, the
exact computation cannot be accelerated; thus, attention has turned towards
tractable approximation schemes.  Two simple schemes for the approximation of
$f^*(p_q)$ are well-known: {\it absolute value approximation} and {\it relative
value approximation}.  Absolute value approximation requires that each density
estimate $f(p_q)$ is within $\epsilon$ of the true estimate $f^*(p_q)$:

\begin{equation}
\label{eqn:ava}
| f(p_q) - f^*(p_q) | < \epsilon \; \; \forall p_q \in S_q.
\end{equation}

Relative value approximation is a more flexible approximation scheme; given some
parameter $\epsilon$, the requirement is that each density estimate is within a
relative tolerance of $f^*(p_q):$

\begin{equation}
\label{eqn:rva}
\frac{| f(p_q) - f^*(p_q) |}{| f^*(p_q) |} < \epsilon \; \; \forall p_q \in S_q.
\end{equation}

Kernel density estimation is related to the well-studied problem of kernel
summation, which can also be solved with dual-tree algorithms
\citep{lee2006faster, lee2008fast}.  In both of those problems, regardless of
the approximation scheme, simple geometric observations can be made to
accelerate computation: when $\mathcal{K}(\cdot, \cdot)$ is shift-invariant,
faraway points have very small kernel evaluations.  Thus, trees can be built on
$S_q$ and $S_r$, and node combinations can be pruned when the nodes are far
apart while still obeying the error bounds.

In the following two subsections, we will separately consider both the absolute
value approximation scheme and the relative value approximation scheme, under
the assumption of a shift-invariant kernel $\mathcal{K}(p_q, p_r) =
\mathcal{K}(\| p_q - p_r \|)$ which is monotonically decreasing and
non-negative.  In addition, we assume that there exists some bandwidth $h$ such
that $\mathcal{K}(d)$ must be concave for $d \in [0, h]$ and convex for $d \in
[h, \infty)$.  This assumption implies that the magnitude of the derivative
$|\mathcal{K}'(d)|$ is maximized at $d = h$.  These are not restrictive
assumptions; most standard kernels fall into this class, including the Gaussian,
exponential, and Epanechnikov kernels.

\subsection{Absolute value approximation}

A tree-independent algorithm for solving approximate kernel density estimation
with absolute value approximation under the previous assumptions on the kernel
is given as a \texttt{BaseCase()} function in Algorithm \ref{alg:kde_base_case}
and a \texttt{Score()} function in Algorithm \ref{alg:kde_score}  \citep[a
correctness proof can be found in][]{curtin2013tree}.  The list $f_p$ holds
partial kernel density estimates for each query point, and the list $f_n$ holds
partial kernel density estimates for each query node.  At the beginning of the
dual-tree traversal, the lists $f_p$ and $f_n$, which are both of size $O(N)$,
are each initialized to 0.  As the traversal proceeds, node combinations are
pruned if the difference between the maximum kernel value
$\mathcal{K}(d_{\min}(\mathscr{N}_q, \mathscr{N}_r))$ and the minimum kernel
value $\mathcal{K}(d_{\max}(\mathscr{N}_q, \mathscr{N}_r))$ is sufficiently
small (line \ref{alg:ava-kde-prune}).  If the node combination can be pruned,
then the partial node estimate is updated (line \ref{alg:ava-kde-update}).  When
node combinations cannot be pruned, \texttt{BaseCase()} may be called, which
simply updates the partial point estimate with the exact kernel evaluation (line
\ref{alg:kde-bc-update}).

After the dual-tree traversal, the actual kernel density estimates $f$
must be extracted.  This can be done by traversing the query tree and
calculating $f(p_q) = f_p(p_q) + \sum_{\mathscr{N}_i \in T}
f_n(\mathscr{N}_i)$, where $T$ is the set of nodes in $\mathscr{T}_q$ that
have $p_q$ as a descendant.  
Each query node needs to be visited only once to perform this calculation; it
may therefore be accomplished in $O(N)$ time.

Note that this version is far simpler than other dual-tree algorithms that have
been proposed for approximate kernel density estimation \citep[see, for
instance][]{gray2003nonparametric}; however, this version is sufficient for our
runtime analysis.  Real-world implementations, such as the one found in
\textbf{mlpack} \citep{mlpack2013}, tend to be far more complex.

\begin{algorithm}[tb]
  \begin{algorithmic}[1]
    \STATE {\bfseries Input:} query point $p_q$, reference point $p_r$, list of
  kernel point estimates $\hat{f}_p$
    \STATE {\bfseries Output:} kernel value $\mathcal{K}(p_q, p_r)$

    \medskip
    \STATE $f_p(p_q) \gets f_p(p_q) + \mathcal{K}(p_q, p_r)$
\label{alg:kde-bc-update}
    \RETURN $\mathcal{K}(p_q, p_r)$
  \end{algorithmic}

  \caption{Approximate kernel density estimation \texttt{BaseCase()}}
  \label{alg:kde_base_case}
\end{algorithm}

\begin{algorithm}[tb]
  \begin{algorithmic}[1]
    \STATE {\bfseries Input:} query node $\mathscr{N}_q$, reference node
$\mathscr{N}_r$, list of node kernel estimates $\hat{f}_n$
    \STATE {\bfseries Output:} a score for the node combination $(\mathscr{N}_q,
\mathscr{N}_r)$, or $\infty$ if the combination should be pruned

    \medskip

    \IF{$\mathcal{K}(d_{\min}(\mathscr{N}_q, \mathscr{N}_r)) -
\mathcal{K}(d_{\max}(\mathscr{N}_q, \mathscr{N}_r)) < \epsilon$}
\label{alg:ava-kde-prune}
      \STATE $f_n(\mathscr{N}_q) \gets f_n(\mathscr{N}_q) + |
\mathscr{D}^p(\mathscr{N}_r) | \left(\mathcal{K}(d_{\min}(\mathscr{N}_q,
\mathscr{N}_r)) + \mathcal{K}(d_{\max}(\mathscr{N}_q, \mathscr{N}_r))\right)
/\;2$ \label{alg:ava-kde-update}
      \RETURN $\infty$
    \ENDIF

    \RETURN $\mathcal{K}(d_{\min}(\mathscr{N}_q, \mathscr{N}_r)) -
\mathcal{K}(d_{\max}(\mathscr{N}_q, \mathscr{N}_r))$
  \end{algorithmic}

  \caption{Absolute-value approximate kernel density estimation
\texttt{Score()}}
  \label{alg:kde_score}
\end{algorithm}

\begin{thm}
Assume that $\mathcal{K}(\cdot, \cdot)$ is a kernel satisfying the assumptions
of the previous subsection.  Then, given a query set $S_q$ and a reference set
$S_r$ with expansion constant $c_r$, and using the approximate kernel density
estimation \texttt{BaseCase()} and \texttt{Score()} as given in Algorithms
\ref{alg:kde_base_case} and \ref{alg:kde_score}, respectively, with the
traversal given in Algorithm \ref{alg:cover-tree-dual}, the running time of
approximate kernel density estimation for some error parameter $\epsilon$ is
bounded by
$O(c_r^{8 + \lceil \log_2 \zeta \rceil} (N + i_t(\mathscr{T}_q) + \theta))$
with $\zeta = -\mathcal{K}'(h) \mathcal{K}^{-1}(\epsilon) \epsilon^{-1}$,
$i_t(\mathscr{T}_q)$ defined as in Definition \ref{def:imbalance}, and $\theta$
defined as in Lemma \ref{lem:extcase3}.

\label{thm:kde-bound}
\end{thm}

\begin{proof}
It is clear that \texttt{BaseCase()} and \texttt{Score()} both take $O(1)$ time,
so Theorem \ref{thm:ct-runtime} implies the total runtime of the dual-tree
algorithm is bounded by $O(c_r^4 |R^*| (N + i_t(\mathscr{T}_q) + \theta))$.
Thus, we will bound $|R^*|$ using techniques related to those used by
\citet{ram2009}.  The bounding of $|R^*|$ is split into two sections: first,
we show that when the scale $s_r^{\max}$ is small enough, $R^*$ is empty.
Second, we bound $R^*$ when $s_r^{\max}$ is larger.

The \texttt{Score()} function is such that any node in $R^*$ for a given query
node $\mathscr{N}_q$ obeys

\begin{equation}
\mathcal{K}(d_{\min}(\mathscr{N}_q, \mathscr{N}_r)) -
\mathcal{K}(d_{\max}(\mathscr{N}_q, \mathscr{N}_r))
\ge \epsilon.
\end{equation}

Thus, we are interested in the maximum possible value $\mathcal{K}(a) -
\mathcal{K}(b)$ for a fixed value of $b - a > 0$.  Due to our assumptions, the
maximum value of $\mathcal{K}'(\cdot)$ is
$\mathcal{K}'(h)$; therefore, the maximum possible value of $\mathcal{K}(a) -
\mathcal{K}(b)$ is when the interval $[a, b]$ is centered on $h$.  This allows
us to say that $\mathcal{K}(a) - \mathcal{K}(b) \le \epsilon$ when $(b - a) \le
(-\epsilon / \mathcal{K}'(h))$.  Note that

\begin{eqnarray}
d_{\max}(\mathscr{N}_q, \mathscr{N}_r) - d_{\min}(\mathscr{N}_q, \mathscr{N}_r)
&\le& d(p_q, p_r) + 2^{s_r^{\max} + 1} - d(p_q, p_r) + 2^{s_r^{\max} + 1} \\
 &\le& 2^{s_r^{\max} + 2}.
\end{eqnarray}

Therefore, $R^* = \emptyset$ when
$2^{s_r^{\max} + 2} \le -\epsilon / \mathcal{K}'(h)$, or when
$s_r^{\max} \le \log_2( -\epsilon / \mathcal{K}'(h) ) - 2$.
%
Consider, then, the case when $s_r^{\max} > \log_2( -\epsilon /
\mathcal{K}'(h) ) - 2$.  Because of the pruning rule, for any $\mathscr{N}_r \in
R^*$, $\mathcal{K}(d_{\min}(\mathscr{N}_q, \mathscr{N}_r)) > \epsilon$;
we may refactor this by applying definitions to show
$d(p_q, p_r) < \mathcal{K}^{-1}(\epsilon) + 2^{s_r^{\max} + 1}$.
Therefore, bounding the number of points in the set
$B_{S_r}(p_q, \mathcal{K}^{-1}(\epsilon) + 2^{s_r^{\max} + 1}) \cap
C_{s_r^{\max}}$
is sufficient to bound $|R^*|$.  For notational convenience, define $\omega =
(\mathcal{K}^{-1}(\epsilon) / 2^{s_r^{\max} + 1}) + 1$, and the statement may be
more concisely written as $B_{S_r}(p_q, \omega 2^{s_r^{\max} + 1}) \cap
C_{s_r^{\max}}$.

Using Lemma \ref{lem:packing} with $\delta = 2^{s_r^{\max}}$ and $\rho = 2
\omega$ gives $|R^*| = c_r^{3 + \lceil \log_2 \omega \rceil}$.







The value $\omega$ is maximized when $s_r^{\max}$ is minimized.  Using the lower
bound on $s_r^{\max}$, $\omega$ is bounded as
%
$\omega = -2 \mathcal{K}'(h) \mathcal{K}^{-1}(\epsilon) \epsilon^{-1}$.
Finally, with $\zeta = -\mathcal{K}'(h) \mathcal{K}^{-1}(\epsilon)
\epsilon^{-1}$, we are able to conclude that $|R^*| \le c_r^{3 + \lceil \log_2
(2 \zeta ) \rceil} = c_r^{4 + \lceil \log_2 \zeta \rceil}$.  Therefore, the
entire dual-tree traversal takes $O(c_r^{8 + \lceil \log_2 \zeta \rceil} (N +
\theta))$ time.

The postprocessing step to extract the estimates $f(\cdot)$ requires one
traversal of the tree $\mathscr{T}_r$; the tree has $O(N)$ nodes, so this takes
only $O(N)$ time.
This is less than the runtime of the
dual-tree traversal, so the runtime of the dual-tree traversal dominates the
algorithm's runtime, and the theorem holds.
\end{proof}

The dependence on $\epsilon$ (through $\zeta$) is expected: as $\epsilon \to 0$
and the search becomes exact, $\zeta$ diverges both because $\epsilon^{-1}$
diverges and also because $\mathcal{K}^{-1}(\epsilon)$ diverges, and the runtime
goes to the worst-case $O(N^2)$; exact kernel density estimation means no nodes
can be pruned at all.



For the Gaussian kernel with bandwidth $\sigma$ defined by $\mathcal{K}_g(d) =
\exp(-d^2 / (2 \sigma^2))$, $\zeta$ does not
depend on the kernel bandwidth; only
the approximation parameter $\epsilon$.  For this kernel, $h = \sigma$ and
therefore $-\mathcal{K}'_g(h) = \sigma^{-1} e^{-1 / 2}$.  Additionally,
$\mathcal{K}_g^{-1}(\epsilon) = \sigma \sqrt{2 \ln (1 / \epsilon)}$.  This means
that for the Gaussian kernel, $\zeta = \sqrt{(-2 \ln \epsilon) / (e
\epsilon^2)}$.  Again, as $\epsilon \to 0$, the
runtime diverges; however, note that there is no dependence on the kernel
bandwidth $\sigma$.  To demonstrate the relationship of runtime to $\epsilon$,
see that for a reasonably chosen $\epsilon = 0.05$, the runtime is approximately
$O(c_r^{8.89} (N + \theta))$; for $\epsilon = 0.01$, the runtime is
approximately $O(c_r^{11.52} (N + \theta))$.  For very small $\epsilon =
0.00001$, the runtime is approximately $O(c_r^{22.15} (N + \theta))$.

Next, consider the exponential kernel:
$\mathcal{K}_l(d) = \exp(-d / \sigma)$.  For this kernel, $h = 0$ (that is, the
kernel is always convex), so then $\mathcal{K}'_l(h) = \sigma^{-1}$.  Simple
algebraic manipulation gives $\mathcal{K}^{-1}_l(\epsilon) = -\sigma \ln
\epsilon$, resulting in $\zeta = -\mathcal{K}'_l(h) \mathcal{K}^{-1}_l(\epsilon)
\epsilon^{-1} = \epsilon^{-1} \ln \epsilon$.  So both the exponential and Gaussian
kernels do not exhibit dependence on the bandwidth.

To understand the lack of dependence on kernel bandwidth more intuitively,
consider that as the kernel bandwidth increases, two things happen: {\it (a)}
the reference set $R$ becomes empty at larger scales, and {\it (b)}
$\mathcal{K}^{-1}(\epsilon)$ grows, allowing less pruning at higher levels.
These effects are opposite, and for the Gaussian and exponential kernels they
cancel each other out, giving the same bound regardless of bandwidth.

\subsection{Relative Value Approximation}

Approximate kernel density estimation using relative-value approximation may be
bounded by reducing the absolute-value approximation algorithm (in linear time
or less) to relative-value approximation.  This is the same strategy as
performed by \citet{ram2009}.

First, we must establish a \texttt{Score()} function for relative value
approximation.  The difference between Equations \ref{eqn:ava} and \ref{eqn:rva}
is the division by the term $|f^*(p_q)|$.  But we can quickly bound
$|f^*(p_q)|$:

\begin{equation}
|f^*(p_q)| \ge N \mathcal{K}\left(\max_{p_r \in S_r} d(p_q, p_r)\right).
\end{equation}

\begin{algorithm}[tb]
  \begin{algorithmic}[1]
    \STATE {\bfseries Input:} query node $\mathscr{N}_q$, reference node
$\mathscr{N}_r$, list of node kernel estimates $\hat{f}_n$
    \STATE {\bfseries Output:} a score for the node combination $(\mathscr{N}_q,
\mathscr{N}_r)$, or $\infty$ if the combination should be pruned

    \medskip

    \IF{$\mathcal{K}(d_{\min}(\mathscr{N}_q, \mathscr{N}_r)) -
\mathcal{K}(d_{\max}(\mathscr{N}_q, \mathscr{N}_r)) < \epsilon
\mathcal{K}^{\max}$}
\label{alg:rva-kde-prune}
      \STATE $f_n(\mathscr{N}_q) \gets f_n(\mathscr{N}_q) + |
\mathscr{D}^p(\mathscr{N}_r) | \left(\mathcal{K}(d_{\min}(\mathscr{N}_q,
\mathscr{N}_r)) + \mathcal{K}(d_{\max}(\mathscr{N}_q, \mathscr{N}_r))\right)
/\;2$ \label{alg:rva-kde-update}
      \RETURN $\infty$
    \ENDIF

    \RETURN $\mathcal{K}(d_{\min}(\mathscr{N}_q, \mathscr{N}_r)) -
\mathcal{K}(d_{\max}(\mathscr{N}_q, \mathscr{N}_r))$
  \end{algorithmic}

  \caption{Relative-value approximate kernel density estimation
\texttt{Score()}}
  \label{alg:kde_rva_score}
\end{algorithm}

This is clearly true: each point in $S_r$ must contribute more than
$\mathcal{K}(\max_{p_r \in S_r} d(p_q, p_r))$ to $f^*(p_q)$.  Now, we may revise
the relative approximation condition in Equation \ref{eqn:rva}:

\begin{equation}
| f(p_q) - f^*(p_q) | \le \epsilon \mathcal{K}^{\max}
\end{equation}

\noindent where $\mathcal{K}^{\max}$ is lower bounded by $\mathcal{K}(\max_{p_r
\in S_r} d(p_q, p_r))$.  Assuming we have some estimate $\mathcal{K}^{\max}$,
this allows us to create a \texttt{Score()} algorithm, given in Algorithm
\ref{alg:kde_rva_score}.

Using this, we may prove linear runtime bounds for relative value approximate
kernel density estimation.

\begin{thm}
Assume that $\mathcal{K}(\cdot, \cdot)$ is a kernel satisfying the same
assumptions as Theorem \ref{thm:kde-bound}.  Then, given a query set $S_q$ and a
reference set $S_r$ both of size $O(N)$, it is possible to perform relative
value approximate kernel density estimation (satisfying the condition of
Equation \ref{eqn:rva}) in $O(N)$ time, assuming that the expansion constant
$c_r$ of $S_r$ is not dependent on $N$.
\end{thm}

\begin{proof}
It is easy to see that Theorem \ref{thm:kde-bound} may be adapted to the very
slightly different \texttt{Score()} rule of Algorithm \ref{alg:kde_rva_score}
while still providing an $O(N)$ bound.  With that \texttt{Score()} function, the
dual-tree algorithm will return relative-value approximate kernel density
estimates satisfying Equation \ref{eqn:rva}.

We now turn to the calculation of $\mathcal{K}^{\max}$.  Given the cover trees
$\mathscr{T}_q$ and $\mathscr{T}_r$ with root nodes $\mathscr{N}_{r}^{R}$ and
$\mathscr{N}_{r}^{R}$, respectively, we may calculate a suitable
$\mathcal{K}^{\max}$ value in constant time:

\begin{equation}
\mathcal{K}^{\max} = d_{\max}(\mathscr{N}_q^R, \mathscr{N}_r^R) = d(p_q^R,
p_r^R) + 2^{s_q^{\max} + 1} + 2^{s_r^{\max} + 1}.
\end{equation}

This proves the statement of the theorem.
\end{proof}

In this case, we have not shown tighter bounds because the algorithm we have
proposed is not useful in practice.  For an example of a better relative-value
approximate kernel density estimation dual-tree algorithm, see the work of
\citet{gray2003nonparametric}.

\section{Range search and range count}
\label{sec:rs}

In the range search problem, the task is to find the set of reference points

\begin{equation}
S[p_q] = \{ p_r \in S_r : d(p_q, p_r) \in [l, u] \}
\end{equation}

\noindent for each query point $p_q$, where $[l, u]$ is the given
range.  The range count problem is practically identical, but only the size of
the set, $|S[p_q]|$, is desired.
Our proof works for both of these algorithms
similarly, but we will focus on range search.  A \texttt{BaseCase()} and
\texttt{Score()} function are given in Algorithms \ref{alg:rs_bc} and
\ref{alg:rs_sc}, respectively \citep[a correctness proof can be found
in][]{curtin2013tree}.  The sets $N[p_q]$ (for each $p_q$) are
initialized to $\emptyset$ at the beginning of the traversal.

In order to bound the running time of dual-tree range search, we require better
notions for understanding the difficulty of the problem.  Observe that if the
range is sufficiently large, then for every query point $p_q$, $S[p_q] = S_r$.
Clearly, for $S_q \sim S_r \sim O(N)$, this cannot be solved in anything less
than quadratic time simply due to the time required to fill each output array
$S[p_q]$.  Define the maximum result size for a given query set $S_q$, reference
set $S_r$, and range $[l, u]$ as

\begin{equation}
| S_{\max} | = \max_{p_q \in S_q} | S[p_q] |.
\label{eqn:smax}
\end{equation}

\begin{algorithm}[tb]
\begin{algorithmic}[1]
    \STATE {\bfseries Input:} query point $p_q$, reference point $p_r$, range
sets $N[p_q]$ and range $[l, u]$
    \STATE {\bfseries Output:} distance $d$ between $p_q$ and $p_r$


    \IF{$d(p_q, p_r) \in [r_{\min}, r_{\max}]$ \AND \texttt{BaseCase($p_q$, $p_r$)} not yet called}
    \STATE  $S[p_q] \gets S[p_q] \cup \{ p_r \}$
    \ENDIF

    \RETURN $d$
  \end{algorithmic}
  \caption{Range search \texttt{BaseCase()}}
  \label{alg:rs_bc}
\end{algorithm}

\begin{algorithm}[tb]
  \begin{algorithmic}[1]
    \STATE {\bfseries Input:} query node $\mathscr{N}_q$, reference node
$\mathscr{N}_r$
    \STATE {\bfseries Output:} a score for the node combination $(\mathscr{N}_q,
\mathscr{N}_r)$, or $\infty$ if the combination should be pruned

    \medskip

    \IF{$d_{\min}(\mathscr{N}_q, \mathscr{N}_r) \in [l, u]$ or
$d_{\max}(\mathscr{N}_q, \mathscr{N}_r) \in [l, u]$}
      \RETURN $d_{\min}(\mathscr{N}_q, \mathscr{N}_r)$
    \ENDIF

    \RETURN $\infty$
  \end{algorithmic}
  \caption{Range search \texttt{Score()}}
  \label{alg:rs_sc}
\end{algorithm}

Small $| S_{\max} |$ implies an easy problem; large $| S_{\max} |$ implies a
difficult problem.  For bounding the running time of range search, we require
one more notion of difficulty, related to how $| S_{\max} |$ changes due to
changes in the range $[l, u]$.

\begin{defn}
For a range search problem with query set $S_q$, reference set $S_r$, range $[l,
u]$, and results $S[p_q]$ for each query point $p_q$ given as

\begin{equation}
S[p_q] = \{ p_r : p_r \in S_r, l \le d(p_q, p_r) \le u \},
\end{equation}

\noindent define the {\it $\alpha$-expansion} of the range set $S[p_q]$ as the slightly larger set

\begin{equation}
S^{\alpha}[p_q] = \{ p_r : p_r \in S_r, (1 - \alpha) l \le d(p_q, p_r) \le (1 +
\alpha) u \}.
\end{equation}
\end{defn}

When the $\alpha$-expansion of the set $S_{\max}$ is approximately the same size
as $S_{\max}$, then the problem would not be significantly more difficult if the
range $[l, u]$ was increased slightly.  Using these notions, then, we may now
bound the running time of range search.

\begin{thm}
Given a reference set $S_r$ of size $N$ with expansion constant $c_r$, and a
query set $S_q$ of size $O(N)$, a search range of $[l, u]$, and using the range
search \texttt{BaseCase()} and \texttt{Score()} as given in Algorithms
\ref{alg:rs_bc} and \ref{alg:rs_sc}, respectively, with the standard cover tree
pruning dual-tree traversal as given in Algorithm \ref{alg:cover-tree-dual}, and
also assuming that for some $\alpha > 0$,

\begin{equation}
| S^{\alpha}[p_q] \setminus S[p_q] | \le C \; \; \forall \; p_q \in S_q,
\end{equation}


\noindent the running time of range search or range count is bounded by

\begin{equation}
O\left(c_r^{4} \max\left(c_r^{4 + \beta}, |S_{\max}| + C\right) (N +
i_t(\mathscr{N}_q) + \theta)
\right)
\end{equation}

with $\theta$ defined as in Lemma \ref{lem:extcase3},
$\beta = \lceil \log_2 (1 + \alpha^{-1}) \rceil$, and $S_{\max}$ as defined in
Equation \ref{eqn:smax}.
\label{thm:rs}
\end{thm}

\begin{proof}
Both \texttt{BaseCase()} (Algorithm \ref{alg:rs_bc}) and \texttt{Score()}
(Algorithm \ref{alg:rs_sc}) take $O(1)$ time.  Therefore, using Lemma
\ref{thm:ct-runtime}, we know that the runtime of the algorithm is bounded by
$O(c_r^4 |R^*| (N + i_t(\mathscr{N}_q) + \theta))$.  As with the previous
proofs, then, our only task is to bound the maximum size of the reference set,
$|R^*|$.

By the pruning rule, for a query node $\mathscr{N}_q$, the reference set $R^*$
is made up of reference nodes $\mathscr{N}_r$ that are within a margin of
$2^{s_q + 1} + 2^{s_r + 1} \le 2^{s_r^{\max} + 2}$ of the range $[l, u]$.  Given
that $p_r$ is the point in $\mathscr{N}_r$,

\begin{equation}
p_r \in \left( B_{S_r}(p_q, u + 2^{s_r^{\max} + 2}) \cap
C_{s_r^{\max}}\right)\setminus \left( B_{S_r}(p_q, l - 2^{s_r^{\max} + 2}) \cap
C_{s_r^{\max}} \right). \label{eqn:rsballs}
\end{equation}

A bound on the number of elements in this set is a bound on $|R^*|$.  
First, consider the case where $u \le \alpha^{-1} 2^{s_r^{\max} + 2}$.  Ignoring
the smaller ball, take $\delta = 2^{s_r^{\max}}$ and $\rho = 4 (1 +
\alpha^{-1})$ and apply Lemma \ref{lem:packing} to produce the bound

\begin{equation}
|R^*| \le c_r^{4 + \lceil \log_2(1 + \alpha^{-1}) \rceil}.
\end{equation}

Now, consider the other case: $u > \alpha^{-1} 2^{s_r^{\max} + 1}$.
This means

\begin{equation}
B_{S_r}(p_q, u + 2^{s_r^{\max} + 1}) \setminus B_{S_r}(p_q, l - 2^{s_r^{\max} +
1}) \subseteq B_{S_r}(p_q, (1 + \alpha) u) \setminus B_{S_r}(p_q, (1 - \alpha)
l).
\end{equation}

This set is necessarily a subset of $S^{\alpha}[p_q]$; by
assumption, the number of points in this set is bounded above by $|S_{\max}| +
C$.  We may then conclude that $|R^*| \le |S_{\max}| + C$.  By taking the
maximum of the sizes of $|R^*|$ in both cases above, we obtain the statement of
the theorem.
\end{proof}

This bound displays both the expected dependence on $c_r$ and $|S_{\max}|$.  As
the largest range set $S_{\max}$ increases in size (with the worst case being
$S_{\max} \sim N$), the runtime degenerates to quadratic.  But for adequately
small $S_{\max}$ the runtime is instead dependent on $c_r$ and the parameter $C$
of the $\alpha$-expansion of $S_{\max}$.  This situation leads to a
simplification.

\begin{cor}
For sufficiently small $|S_{\max}|$ and sufficiently small $C$, the runtime of
range search under the conditions of Theorem \ref{thm:rs} simplifies to

\begin{equation}
O(c_r^{8 + \beta} (N + i_t(\mathscr{N}_q) + \theta)).
\end{equation}
\label{cor:rs}
\end{cor}

In this setting we can more easily consider the relation of the running time to
$\alpha$.  Consider $\alpha = (1 / 3)$; this yields a running time of $O(c^8 (N
+ \theta))$.  $\alpha = (1 / 7)$ yields $O(c^9 (N + i_t(\mathscr{N}_q +
\theta))$, $\alpha = (1 / 15)$ yields $O(c^{10} (N + i_t(\mathscr{N}_q) +
\theta))$, and so forth.  As $\alpha$ gets smaller, the exponent on $c$ gets
larger, and diverges as $\alpha \to 0$.

For reasonable runtime it is necessary that the $\alpha$-expansion of $S_{\max}$
be bounded.  This is because the dual-tree recursion must retain reference nodes
which may contain descendants in the range set $S[p_q]$ for some query $p_q$.
The parameter $C$ of the $\alpha$-expansion allows us to bound the number of
reference nodes of this type, and if $\alpha$ increases but $C$ remains small
enough that Corollary \ref{cor:rs} applies, then we are able to obtain tighter
running bounds.

\section{Conclusion}
\label{sec:conclusion}

We have presented a unified framework for bounding the runtimes of dual-tree
algorithms that use cover trees and the standard cover tree pruning dual-tree
traversal (Algorithm \ref{alg:cover-tree-dual}).  In order to produce an
understandable bound, we have introduced the notion of cover tree imbalance; one
possible interesting direction of future work is to empirically and
theoretically minimize this quantity by way of modified tree construction
algorithms; this is likely to provide both tighter runtime bounds and also
accelerated empirical results.

Our main result, Theorem \ref{thm:ct-runtime}, allows plug-and-play runtime
bounding of these algorithms.  We have shown that Theorem \ref{thm:ct-runtime}
is useful for bounding the runtime of nearest neighbor search (Theorem
\ref{thm:nns}), approximate kernel density estimation (Theorem
\ref{thm:kde-bound}), exact range count, and exact range search (Theorem
\ref{thm:rs}).  With our contribution, bounding a cover tree dual-tree algorithm
is straightforward and only involves bounding the maximum size of the reference
set, $| R^* |$.

\bibliography{bounds}


\end{document}